\documentclass{lmcs} 
\pdfoutput=1
\usepackage[utf8]{inputenc}

\usepackage{lastpage}
\lmcsdoi{19}{4}{8}
\lmcsheading{}{\pageref{LastPage}}{}{}%
{Aug.~16,~2022}{Nov.~03,~2023}{}


\usepackage{amsmath}
\usepackage{amsfonts}
\usepackage{amssymb}
\usepackage{xspace}
\usepackage{graphicx}

\usepackage{multicol}
\usepackage{stmaryrd} 
\RequirePackage[bookmarks,colorlinks=true]{hyperref}
\usepackage[capitalise]{cleveref}

\usepackage{etoolbox}
\usepackage{array}
\usepackage{multirow}
\usepackage{pifont}

\usepackage{tikz}
\usetikzlibrary{arrows,automata,decorations.pathmorphing}
\usetikzlibrary{positioning}
\usetikzlibrary{shapes,shapes.geometric,fit,calc,automata,snakes}

\newcommand{\NotNeeded}[1]{}
\newcommand{\FullVersion}[1]{}
\newcommand{\NotNow}[1]{}
\newcommand{\CheckRelevance}[1]{}

\newcommand{\Subject}[1]{\paragraph*{#1.}}
\newcommand{\St}{~|~}

\newcommand{\Nat}{\ensuremath{\mathbb{N}}\xspace}
\newcommand{\Rat}{\ensuremath{\mathbb{Q}}\xspace}
\newcommand{\Reals}{\ensuremath{\mathbb{R}}\xspace}

\newcommand{\Func}[1]{{\mathsf{#1}}\xspace}

\newcommand{\A}{{\mathcal A}}
\newcommand{\B}{{\mathcal B}}

\newcommand{\D}{{\mathcal D}}

\newcommand{\trans}[3]{#1\xrightarrow[]{#2}#3}
\newcommand{\letter}{\sigma}
\newcommand{\weight}{x}

\newcommand{\Val}{\Func{Val}}
\newcommand{\Inf}{\Func{Inf}}
\newcommand{\Sup}{\Func{Sup}}
\newcommand{\Sum}{\Func{Sum}}
\newcommand{\DSum}{\Func{DSum}}
\newcommand{\Avg}{\Func{Avg}}
\newcommand{\LimInf}{\Func{LimInf}}
\newcommand{\LimSup}{\Func{LimSup}}

\newcommand{\Safety}{\Func{Safety}}
\newcommand{\Reach}{\Func{Reachability}}

\newcommand{\Proc}[1]{{\mathtt{#1}}}
\newcommand{\Polish}{\Proc{Polish}}

\newcommand{\PTime}{{\sc{PTime}}\xspace}
\newcommand{\NPcoNP}{{\sc{NP}$\cap$\sc{coNP}}\xspace}
\newcommand{\NP}{{\sc NP}\xspace}

\definecolor{darkgreen}{rgb}{0.13, 0.55, 0.13}

\newcommand{\LW}[2]{#1\!:\!#2}
\newcommand{\strat}{s}
\renewcommand{\epsilon}{\varepsilon}

\newcommand{\HD}{\textsf{HD}\xspace}
\newcommand{\HDness}{\textsf{HDness}\xspace}

\newcommand{\ML}[1]{\begin{tabular}[c]{@{}c@{}} #1\end{tabular} } 

\newcolumntype{M}{>{\centering\arraybackslash}m{2.5cm}}
\newcommand{\Refer}[1]{{\footnotesize #1}}
\newcommand{\Complexity}[1]{#1}

\newcommand{\Ptime}{\Complexity{PTIME}}
\newcommand{\NiceCheckMark}{\ML{\\[-22pt]{\large \ding{51}\normalsize}\\[-15pt]}}
\newcommand{\NiceCross}{\ML{\\[-22pt]{\large \ding{55}\normalsize}\\[-15pt]}}

\crefname{defi}{Definition}{Definitions} 
\crefname{rem}{Remark}{Remarks} 
\crefname{thm}{Theorem}{Theorems}
\crefname{lem}{Lemma}{Lemmas}
\crefname{prop}{Proposition}{Propositions}
\crefname{cor}{Corrolary}{Corrolaries}
\crefname{propC}{Proposition}{Propositions}

\begin{document}

\title[Quantitative Automata Token Games]{Token Games and History-Deterministic\texorpdfstring{\\}{ }Quantitative Automata}
\titlecomment{{\lsuper*}The present article extends \cite{BL22} with additional results on $\Inf$, $\Sup$ and $\Reach$ automata on infinite words as well as some additional discussion throughout.}
\thanks{Research supported by the Israel Science Foundation grant 2410/22.}	

\author[U.~Boker]{Udi Boker\lmcsorcid{0000-0003-4322-8892}}[a]
\author[K.~Lehtinen]{Karoliina Lehtinen\lmcsorcid{0000-0003-1171-8790}}[b]

\address{Reichman University, Herzliya, Israel}	

\address{CNRS, Marseille-Aix Universit\'e, LIS, Marseille, France}	

\begin{abstract}
A nondeterministic automaton is history-deterministic if its nondeterminism can be resolved by only considering the prefix of the word read so far. Due to their good compositional properties, history-deterministic automata are useful in solving games and synthesis problems.
Deciding whether a given nondeterministic automaton is history-deterministic (the \HDness problem) is generally a difficult task, which can involve an exponential procedure, or even be undecidable, as is the case for example with pushdown automata.
\textit{Token games} provide a \PTime solution to the \HDness problem of B\"uchi and coB\"uchi automata, and it is conjectured that $2$-token games characterise \HDness for all $\omega$-regular automata.

We extend token games to the quantitative setting and analyse their potential to help deciding \HDness of quantitative automata.
In particular, we show that $1$-token games characterise \HDness for all quantitative (and Boolean) automata on finite words, as well as discounted-sum ($\DSum$), $\Inf$ and $\Reach$ automata on infinite words, and that $2$-token games characterise \HDness of $\LimInf$ and $\LimSup$ automata, as well as $\Sup$ automata on infinite words.
Using these characterisations, we provide solutions to the \HDness problem of $\Safety$, $\Reach$, $\Inf$ and $\Sup$ automata on finite and infinite words in \PTime, of $\DSum$ automata on finite and infinite words in {\sc NP}$\cap${\sc co-NP}, of $\LimSup$ automata in quasipolynomial time, and of $\LimInf$ automata in exponential time, where the latter two are only polynomial for automata with a logarithmic number of weights. 
\end{abstract}

\maketitle

\section{Introduction}\label{sec:intro}
\subsection*{History-determinism}
A nondeterministic (quantitative) automaton is called  history-deterministic (\HD)~\cite{Col09,BL21}  if its nondeterministic choices can be resolved by only considering the word read so far, uniformly across possible suffixes (see \cref{fig:THDnotHD} for examples of \HD and non-\HD automata). More precisely, there should be a function (strategy), sometimes called a resolver, that maps the finite prefixes of a word to the transition to be taken at the last letter. The run built in this way must, in the Boolean setting, be accepting whenever the word is in the language of the automaton and, in the more general quantitative setting, attain the value of the automaton on the word (\textit{i.e.}, the supremum of all its runs' values).

History-determinism lies in between determinism and nondeterminism, enjoying in some aspects the best of both worlds: \HD automata are, like deterministic ones, useful for solving games and reactive synthesis \cite{HP06,Col09,HPR16,HPR17,CF19,GJLZ21,BL21}, yet can sometimes be more expressive and/or succinct.  For example, \HD  coB\"uchi and $\LimInf$ automata can be exponentially more succinct than deterministic ones \cite{KS15}, and \HD pushdown automata are both more expressive and at least exponentially more succinct than deterministic ones~\cite{LZ20,GJLZ21}.
In the ($\omega$-)regular setting, history-determinism coincides with good-for-gameness \cite{BL19}, a notion characterised by the compositional properties of the automaton \cite{HP06}, while in the quantitative setting it is stronger~\cite{BL21}. The problem of deciding whether a nondeterministic automaton is \HD is interreducible with deciding the \textit{best-value synthesis} problem of a deterministic automaton of the same type~\cite{FLW20,BL21}. In this quantitative version of the reactive synthesis problem, the system must guarantee a behaviour that matches the value of any global behaviour compatible with the environment's actions. The witness of \HDness corresponds exactly to the solution system of this synthesis problem,
providing another motivation for this line of research.

\subsection*{Deciding history-determinism -- a difficult task}
History-determinism is formally defined by a \emph{letter game} played on the automaton $\A$ between Adam and Eve, where Adam produces an input word $w$, letter by letter, and Eve tries to resolve the nondeterminism in $\A$ so that the resulting run attains $\A$'s value on $w$. Then $\A$ is \HD if Eve has a winning strategy in the letter game on it. 
The difficulty of deciding who wins the letter game stems from its complicated winning condition -- Eve wins if her run has the value of the supremum over all runs of $\A$ on $w$. 

The naive solution is to determinise $\A$ into an automaton $\D$, and consider a game equivalent to the letter game that has a simpler winning condition and whose arena is the product of $\A$ and $\D$ \cite{HP06}. The downside with this approach, however, is that it requires the determinisation of $\A$, which often involves a procedure exponential in the size of $\A$ and sometimes is even impossible due to an expressiveness gap.
Note that deciding whether an automaton is good-for-games, which is closely related to whether it is \HD \cite{BL19,BL21}, is also difficult, as  it requires reasoning about composition with all possible games.

\subsection*{Token games -- a possible aid}
The \textit{one-token-game}, which is closely related to the letter game but easier to solve, was introduced by L{\"{o}}ding in an algorithm for deciding determinisability-by-pruning of a regular automaton, and was formally defined in \cite[Definition 5]{LR13}. It was later generalised by Bagnol and Kuperberg~\cite{BK18} to a \textit{$k$-token-game} on $\omega$-regular automata, for a given $k\in \Nat$, in the course of seeking an easier to decide characterisation of history-determinism. 

In a $k$-token game on an automaton $\A$, denoted by $G_k(\A)$, like in the letter game, Adam generates a word $w$ letter by letter, and Eve builds a run on $w$ by resolving the nondeterminism. In addition, Adam also has to resolve the nondeterminism of $\A$ to build $k$ runs letter-by-letter over $w$. 
The winning condition for Eve in these games on Boolean automata is that either all runs built by Adam are rejecting, or Eve's run is accepting. Such games, as they compare a finite number of pairs of runs, are easier to solve than the letter game. 

Then, to decide \HDness of a class of automata, one can attempt to show that the letter game always has the same winner as a $k$-token game, for some $k$, and solve the $k$-token game. (If Eve wins the letter game then she wins the $k$-token game, for every $k$, by using the same strategy, ignoring Adam's runs. However, it might be that she wins a $k$-token game, taking advantage of her knowledge of how Adam resolves the nondeterminism, but loses the letter game.)

Bagnol and Kuperberg showed in \cite{BK18} that on B\"uchi automata, the letter game and the $2$-token game always have the same winner, and in \cite{BKLS20b}, Boker, Kuperberg, Lehtinen and Skrzypczak extended this result to coB\"uchi automata. In both cases, this allows for a \PTime procedure for deciding \HDness. 
Furthermore, Bagnol and Kuperberg suggested in \cite[Conclusion]{BK18} that $2$-token games might characterise \HDness also for parity automata (and therefore for all $\omega$-regular automata); a conjecture (termed later the $G2$ conjecture) that is still open.

\subsection*{Our contribution} 
We extend token games to the quantitative setting, and use them to decide the \HDness of some quantitative automata.
We define a $k$-token game on a quantitative automaton exactly as on a Boolean one, except that Eve wins if her run has a value at least as high as all of Adam's runs.

We show first, in \cref{sec:G1}, that the $1$-token game, in which Adam just has one run to build, characterises \HDness for all quantitative (and Boolean) automata on finite words, and for $\Safety$, $\Reach$, $\Inf$ and 
discounted-sum ($\DSum$) automata on infinite words. This results in a \PTime decision procedure for checking \HDness of $\Safety$, $\Reach$,  and $\Inf$ automata, and an \NPcoNP procedure for $\DSum$ automata, both on finite and infinite words. Note that the complexity for $\DSum$ automata on finite words was already known~\cite{FLW20}, but on infinite words it was erroneously believed to be \NP-hard \cite[Theorem 6]{HPR16}.

Towards getting the above results, we analyse key properties of value functions of quantitative automata, and show that the $1$-token game characterises \HDness for every $\Val$ automaton, such that $\Val$ is present-focused (\cref{def:PresentFocused}), which is in particular the case for all $\Val$ automata on finite words \cite[Lemma 16]{BL21}, as well as $\DSum$ \cite[Lemma 22]{BL21} and $\Inf$ automata on infinite words.

We then show, in \cref{sec:G2}, that the $2$-token game, in which Adam builds two runs, characterises \HDness for both $\LimSup$ and $\LimInf$ automata. The approach here is more involved: it decomposes the quantitative automaton into a collection of B\"uchi or coB\"uchi automata such that if Eve wins the $2$-token game on the original automaton, she also wins in the component automata. Since the $2$-token game characterises \HD for B\"uchi and coB\"uchi automata, the component automata are then \HD and the witness strategies can be combined with the $2$-token strategy of the original automaton to build a letter-game strategy for Eve. 

The general flow of our approach is illustrated in \cref{fig:Flow}. As a corollary, we obtain that $G_2$ also characterises \HDness for $\Sup$ automata on infinite words.

We further present, in \cref{sec:solveG2}, algorithms to decide the winner of the $2$-token games on $\LimInf$ and $\LimSup$ automata via reductions to solving parity games, and on $\Sup$ automata on infinite words via reduction to coB\"uchi games. The complexity of the procedure for a $\LimSup$ automaton $\A$ is the same as that of solving a parity game of size polynomial in the size of $\A$ with twice as many priorities as there are weights in $\A$, which is in quasipolynomial time. For $\LimInf$ automata the procedure is in exponential time. In both cases, it is in polynomial time if the number of weights is  logarithmic in the automaton size. For $\Sup$ automata, the procedure is always polynomial. These results are summarised in~\cref{tab:ComplexityTable}.

For some variants of the synthesis problem, the complexity of the witness of history-determinism is also of interest (for other variants it is not), as it corresponds to the complexity of the implementation of the solution system \cite[Section 5]{BL21}. 
We give an exponential upper bound to the complexity of the witness for $\LimSup$, $\LimInf$ and $\Sup$ automata, which, for $\LimInf$, is tight. As a corollary, we obtain that \HD $\LimSup$ automata are exactly as expressive as deterministic $\LimSup$ automata
and at most exponentially more succinct.

\subsection*{Related work}
In the $\omega$-regular setting (where \HDness coincides with good-for-gameness), \cite[Section 4]{HP06} provides an exponential scheme for checking \HDness of all $\omega$-regular automata, based on  determinisation and checking fair simulation.
\HDness of B\"uchi automata is resolved, as mentioned above, in \PTime, using $2$-token games \cite{BK18}. The coB\"uchi case is also resolved in \PTime, originally via an indirect usage of ``joker games'' \cite{KS15}, and later by using $2$-token games \cite{BKLS20b}.

In the quantitative setting, deciding \HDness coincides with best-value partial domain synthesis~\cite{FLW20}, $0$-regret synthesis~\cite{HPR17} and, for some value functions, $0$-regret determinisation~\cite{FJLPR17,BL21}. There are procedures to decide \HDness (which is sometimes called good-for-gameness due to erroneously assuming them equivalent) of $\Sum$, $\Avg$, and $\DSum$ automata on finite words, as follows.
For $\Sum$ and $\Avg$ automata on finite words, a \PTime solution combines \cite[Theorem 4.1]{AKL10}, which provides a \PTime algorithm for checking whether such an automaton is ``determinisable by pruning'', and \cite[Theorem 21]{BL21}, which shows that such an automaton is \HD if and only if it is determinisable by pruning.

\begin{prop}
	Deciding whether a $\Sum$ or $\Avg$ automaton on finite words is history-deterministic is in \PTime.
\end{prop}

For $\DSum$ automata on finite words, \cite[Theorem 23]{FLW20} provides an  {\sc NP}$\cap${\sc co-NP} solution, using a game that is quite similar to the $1$-token game, differing from it in a few aspects -- for example, Adam is asked to either copy Eve with his token or move into a second phase where he plays transitions first -- and  uses a characterisation of \HD strategies resembling our notion of cautious strategies (\cref{def:Cautious}) specialised to $\DSum$ automata. 

\begin{table}[h!]
	\begin{center}
		\def\arraystretch{1.1} 
		\resizebox{\textwidth}{!}{
			\begin{tabular}{c||M|M|M|M|}
				\ \multirow{2}{*}{Automata Type~}& \multicolumn{2}{c|}{\HD is characterised by} & \multicolumn{2}{c|}{\HDness Complexity} \\ 
				\cline{2-5}
				\ & $G_1$ & $G_2$& Finite words & Infinite words \\ \hline
				\hline
				All automata on& \multicolumn{2}{c|}{ \NiceCheckMark}&\multicolumn{2}{c|}{\multirow{2}{*}{Varies}} \\
				finite words& \multicolumn{2}{c|}{\Refer{\cref{cl:FiniteWords}}}&\multicolumn{2}{c|}{}\\
				\hline
				
				All automata on& {\bf Not all} & {\bf Open} &\multicolumn{2}{c|}{\multirow{2}{*}{Varies}} \\
				infinite words& \Refer{\cite[Lemma 8]{BK18}}&\cite{BK18}&\multicolumn{2}{c|}{}\\
				\hline
				
				\multirow{2}{*}{$\Safety$} &  \multicolumn{2}{c|}{ \NiceCheckMark}   & \multicolumn{2}{c|}{\Complexity{\Ptime}} \\
				& \multicolumn{2}{c|}{\Refer{\cref{cl:InfIsG1}}} &\Refer{\cref{cl:solveG1-reach-safe}} & \Refer{\cref{cl:solveG1-reach-safe-infinite-words}}   \\
				\hline
				
				\multirow{2}{*}{$\Reach$} &  \multicolumn{2}{c|}{ \NiceCheckMark}   & \multicolumn{2}{c|}{\Complexity{\Ptime}} \\
				& \multicolumn{2}{c|}{\Refer{\cref{cl:Reachability}}} &\Refer{\cref{cl:solveG1-reach-safe}} & \Refer{\cref{cl:solveG1-reach-safe-infinite-words}}   \\
				\hline

				\multirow{2}{*}{$\Inf$} &  \multicolumn{2}{c|}{ \NiceCheckMark}   & \multicolumn{2}{c|}{\Complexity{\Ptime}} \\
				& \multicolumn{2}{c|}{\Refer{\cref{cl:InfIsG1}}} &\multicolumn{2}{c|}{\Refer{\cref{cl:InfPtimeInfiniteWords}} }  \\
				\hline
				
				\multirow{2}{*}{$\Sup$} &  \NiceCross& \NiceCheckMark   & \multicolumn{2}{c|}{\Complexity{\Ptime}} \\
				& \Refer{\cref{cl:SupNotG1}} & \Refer{\cref{cl:G2-oSup}} & \Refer{\cref{cl:FiniteSupPtime}}&\Refer{\cref{cl:solveG2-Sup}}   \\
				\hline
				
				\multirow{2}{*}{$\DSum$} &  \multicolumn{2}{c|}{ \NiceCheckMark}   & \multicolumn{2}{c|}{\Complexity{{\sc NP}$\cap${\sc co-NP}}} \\
				& \multicolumn{2}{c|}{\Refer{\cref{cl:DSumG1}}} &\multicolumn{2}{c|}{\Refer{\cref{cl:DSum-NPcoNP}}}   \\
				\hline
				
				\multirow{2}{*}{$\LimInf$} &  \NiceCross& \NiceCheckMark      & \multirow{2}{*}{-} & Quasipoly. \\
				& \Refer{\cite[Lemma 8]{BK18}} &\Refer{\cref{cl:LimInfSupG2}} &&\Refer{\cref{cl:solveG2liminf}}   \\
				\hline
				
				\multirow{2}{*}{$\LimSup$} &  \NiceCross& \NiceCheckMark    & \multirow{2}{*}{-}& Quasipoly.\\
				& \Refer{\cite[Lemma 8]{BK18}}&\Refer{\cref{cl:LimInfSupG2}}&&\Refer{\cref{cl:solveG2LimSup}}  \\
				\hline

			\end{tabular}
		} 
	\end{center}
	\caption{Characterisation of history-determinism by 1- and 2-token games (the characterisation of the specific automata types refers to automata on infinite words), and the complexity of checking whether an automaton is history-deterministic.}
	\label{tab:ComplexityTable}
\end{table}

\section{Preliminaries}\label{sec:Preliminaries}

\Subject{Words}
An \emph{alphabet} $\Sigma$ is a finite nonempty set of letters. A finite (resp.\ infinite) \emph{word} $u=\sigma_0 \ldots \sigma_k\in \Sigma^{*}$ (resp.\ $w=\sigma_0 \sigma_1\ldots\in \Sigma^{\omega}$) is a finite (resp.\ infinite) sequence of letters from $\Sigma$; $\epsilon$ is the empty word.  
We write $\Sigma^\infty$ for $\Sigma^* \cup \Sigma^\omega$.
We use $[i..j]$ to denote a set $\{i,\ldots,j\}$ of integers, $[i]$ for $[i..i]$, $[..j]$ for $[0..j]$, and $[i..]$ for integers equal to or larger than $i$. We write $w[i..j], w[..j]$, and $w[i..]$ for the infix $\sigma_i \ldots \sigma_j$, prefix $\sigma_0 \ldots \sigma_j$, and suffix $\sigma_i \ldots$ of $w$.
A \emph{language} is a set of words.

\Subject{Games}

We consider a variety of turn-based zero-sum games between Adam (A) and Eve (E). Formally, a game is played on an arena of which the positions are partitioned between the two players. A play is a maximal (finite or infinite) path. The winning condition partitions plays into those that are winning for each player. In some of the technical developments we use \textit{parity games}, in which moves are coloured with integer priorities and a play is winning for Eve if the maximal priority that occurs infinitely often along the play is even. A coB\"uchi game is the special case of a parity game with  priorities $1$ and $0$. A weak game is the special case of a coB\"uchi game in which priorities $0$ and $1$ do not both occur within a single cycle.

A strategy for a player $P\in \{A,E\}$ maps partial plays ending in a position belonging to $P$ to a successor position. A (partial) play $\pi$ agrees with a strategy $\strat_P$ of $P$, written $\pi\in \strat_P$, if whenever its prefix $p$ ends in a position of $P$, the next move is $\strat_P(p)$. A strategy of $P$ is winning from a position $v$ if all plays starting at $v$ that agree with it are winning for $P$. A strategy is positional if it maps all plays that end in the same position to the same successor. 

\Subject{Quantitative Automata}

A \emph{nondeterministic quantitative\footnote{We speak of ``quantitative'' rather than ``weighted'' automata, following the distinction made in \cite{Bok21} between the two.} automaton} (or just automaton from here on) on words is a tuple $\A=(\Sigma,Q,\iota,\delta)$, where $\Sigma$ is an alphabet; $Q$ is a finite nonempty set of states; $\iota\in Q$ is an initial state; and $\delta\colon Q\times \Sigma \to 2^{(\Rat \times Q)}$ is a transition function over weight-state pairs. 

A \emph{transition} is a tuple $(q,\letter,\weight,q')\in Q{\times}\Sigma{\times} \Rat\times Q$,  also written $\trans{q}{\letter:\weight }{q'}$. (There might be several transitions with different weights over the same letter between the same states.)
We write $\gamma(t)=\weight$ for the weight of a transition $t=(q,\letter,\weight,q')$.
$\A$ is  deterministic if for all $q\in Q$ and $a\in \Sigma$, $\delta(q,a)$ is a singleton. 
We require that the automaton $\A$ is $\emph{total}$, namely that for every state $q\in Q$ and letter $\letter\in\Sigma$, there is at least one state $q'$ and a transition $\trans{q}{\letter:\weight}{q'}$.

A run of $\A$ on a word $w$ is a sequence $\rho = \trans{q_0}{w[0]:\weight_0}{q_1}\trans{}{w[1]:\weight_1}{q_2}\ldots$ of transitions where $q_0=\iota$ and $(\weight_i,q_{i+1})\in \delta(q_i,w[i])$. As each transition $t_i$ carries a weight $\gamma(t_i)\in\Rat$, the sequence $\rho$ provides a weight sequence $\gamma(\rho) = \gamma(t_0) \gamma(t_1) \ldots$ 
A $\Val$ (e.g., $\Sum$) automaton  is one equipped with a \emph{value function} $\Val:\Rat^* \to \Reals$ or $\Val:\Rat^\omega \to \Reals$, which assigns real values to runs of $\A$. 
The value of a run $\rho$ is $\Val(\gamma(\rho))$. 
The value of $\A$ on a word $w$ is the supremum of $\Val(\rho)$ over all runs $\rho$ of $\A$ on $w$.
Two automata $\A$ and $\A'$ are \emph{equivalent}, if they realise the same function from words to reals. The size of an automaton consists of the maximum among the size of its alphabet, state-space, and transition-space.

\Subject{Value functions}\
We list below the value functions that we will consider in the sequel.

\vspace{3pt}
\noindent For finite sequences $v_0 v_1 \ldots v_{n-1}$ of rational weights:
\vspace{-.3cm}
\begin{multicols}{2}
	\begin{itemize}
		\item $\displaystyle \Sum(v) = \sum_{i=0}^{n-1} v_i$
		\item $\displaystyle \Avg(v) = \frac{1}{n} \sum_{i=0}^{n-1} v_i$
	\end{itemize}
\end{multicols}

\noindent For finite and infinite sequences $v_0 v_1 \ldots$ of rational weights:
\vspace{-.3cm}
\begin{multicols}{2}
	\begin{itemize}
		\item $\displaystyle \Inf(v) = \inf\{v_n \St n \geq 0\}$
		\item $\displaystyle \Sup(v) = \sup\{v_n \St n \geq 0\}$
	\end{itemize}
\end{multicols}
\vspace{-.3cm}
\begin{itemize}
	\item For a discount factor $\lambda\in\Rat\cap(0,1)$, $\lambda$-$\displaystyle \DSum(v) = \sum_{i\geq 0} \lambda^i  v_i$
\end{itemize}

\noindent  For infinite sequences $v_0 v_1 \ldots$ of rational weights:
\vspace{-.2cm}
\begin{multicols}{2}
	\begin{itemize}
		\item $\displaystyle \LimInf(v) = \lim_{n\to\infty}\limits\inf\{v_i \St i \geq n\}$
		\item $\displaystyle \LimSup(v) = \lim_{n\to\infty}\limits\sup\{v_i \St i \geq n\}$
	\end{itemize}
\end{multicols}

\noindent\textit{Regular and $\omega$-regular automata} (with acceptance on transitions) can be viewed as special cases of quantitative automata with weights in $\{0,1\}$, where the language of the automaton consists of words with value $1$. In particular, considering only weights in $\{0,1\}$, B\"uchi coincides with $\LimSup$, coB\"uchi with $\LimInf$, $\Reach$ with $\Sup$ and $\Safety$ with $\Inf$. With this in mind, $\Reach$ and $\Safety$ automata on finite and infinite words, are defined as $\A=(\Sigma,Q,\iota,\delta)$ as above, with weights $0$ and $1$ on transitions, and we assume that for $\Reach$ automata every accepting transition, that is a transition with weight $1$, leads to a sink with self-loops of weight $1$, called the target (we use this assumption in \cref{sec:G1Characterizes}), and for $\Safety$ automata, every rejecting transition, that is a transition with weight $0$, leads to a sink with self-loops of weight $0$. We call the rest of the automaton its safe region. We say that the automaton accepts a word if its value is $1$: A $\Reach$ automaton accepts words with runs that reach the target, while a $\Safety$ automaton accepts words with runs that remain in the safe region. See more on $\omega$-regular automata, e.g., in \cite{Bok18}. 

\Subject{History-determinism}
Intuitively, an automaton is history-deterministic if there is a strategy to resolve its nondeterminism according to the word read so far such that for every word, the value of the resulting run is the value of the word.
\begin{defi}[History-determinism \cite{Col09,BL21}]  \label{def:HistoryDet}
	A $\Val$ automaton $\A$ is \emph{history-deterministic} (\HD) if Eve wins the following win-lose \emph{letter game}, in which Adam chooses the next letter and Eve resolves the nondeterminism, aiming to construct a run whose value is equivalent to the generated word's value.
	\begin{description}
		\item[Letter game] 
		A play begins in $q_0=\iota$ (the initial state of $\A$) and at the $i^{th}$ turn, from state $q_i$, it progresses to a next state $q_{i+1}$ as follows:
		\begin{itemize}
			\item  Adam picks a letter $\sigma_{i}$ from $\Sigma$ and then
			\item  Eve chooses a transition $t_i=\trans{q_{i}}{\sigma_{i}:\weight_{i}}{q_{i+1}}$.
		\end{itemize}
		In the limit, a play consists of an infinite word $w$ that is derived from the concatenation of $\sigma_0,\sigma_1,\ldots$, as well as an infinite sequence $\pi = t_0,t_1,\ldots$ of transitions.
		For $\A$ on infinite words, Eve wins a play in the letter-game if $\Val(\pi) \geq \A(w)$. 
		For $\A$ on finite words, Eve wins if  for all $i\in\Nat$, $\Val(\pi[0..i]) \geq \A(w[0..i])$. 
	\end{description}
\end{defi}

Consider for example the $\LimSup$ automaton $\A$ in \cref{fig:THDnotHD}. Eve loses the letter game on $\A$: 
Adam can start with the letter $a$; then if Eve goes from $s_0$ to $s_1$, Adam continues to choose $a$ forever, generating the word $a^\omega$, where $\A(a^\omega)=3$, while Eve's run has the value $2$. If, on the other hand, Eve chooses on her first move to go from $s_0$ to $s_2$, Adam continues with choosing $b$ forever, generating the word $ab^\omega$, where $\A(ab^\omega)=2$, while Eve's run has the value $1$.

\Subject{Families of value functions}

We will provide some of our results with respect to a family of $\Val$ automata based on properties of the value function $\Val$. 

We first define \textit{cautious 
	strategies} for Eve in both the letter game and token games (\cref{sec:TokenGames}), which we use to define \textit{present-focused} value functions.
Intuitively, a strategy is cautious if it avoids mistakes: it only builds run prefixes that can achieve the maximal value of any continuation of the current word prefix. 

\begin{defi}[Cautious strategies \cite{BL21}]\label{def:Cautious}
	Consider the letter game on a $\Val$ automaton $\A$, in which Eve builds a run of $\A$ transition by transition.
	A move (transition) $t=q\xrightarrow{\sigma:\weight}q'$ of Eve, played after some run $\rho$ ending in a state $q$, is \emph{non-cautious}
	if for some word $w$, there is a run $\pi'$ from $q$ over $\sigma w$ such that $\Val(\rho\pi')$ is strictly greater than the value of $\Val(\rho\pi)$ for any $\pi$ starting with $t$.
	A strategy is \emph{cautious} if it makes no non-cautious moves.
\end{defi}

A winning strategy for Eve in the letter game must of course be cautious; Whether all cautious strategies are winning depends on the value function. For example, a cautious strategy in a $\Safety$ automaton is obviously winning, as it inevitably remains in the safe region while Adam plays prefixes of words in the language, while a cautious strategy in a $\Reach$ automaton might not be winning as cautiousness does not require Eve to ever reach the target with her run. We call a value function \emph{present-focused} if, morally, it depends on the prefixes of the value sequence, formalised by winning the letter game via cautious strategies.

\begin{defi}[Present-focused value functions \cite{BL21}]\label{def:PresentFocused}
	A value function $\Val$, on finite or infinite sequences, is \emph{present-focused} if for all automata $\A$ with value function $\Val$, every cautious strategy for Eve in the letter game on $\A$ is also a winning strategy in that game.
\end{defi}

Value functions on finite sequences are present-focused, as they can only depend on prefixes, while value functions on infinite sequences are not necessarily present-focused \cite[Remark 17]{BL21}, for example $\LimInf$ and $\LimSup$.
\begin{propC}[{\cite[Lemma 16]{BL21}}]\label{cl:finite-is-present-focused}
	Every value function  $\Val$ on finite sequences of rational values is present focused.
\end{propC}

\begin{propC}[{\cite[Lemma 22]{BL21}}]\label{cl:DSumPresentFocused}
	For every $\lambda\in\Rat\cap(0,1)$, $\lambda$-$\DSum$ on infinite sequences of rational values is a present-focused value function.
\end{propC}

\section{Token Games}\label{sec:TokenGames}
The \textit{one-token-game} was introduced by L{\"{o}}ding and was then formally defined in \cite[Definition 5]{LR13}. It was later generalised by Bagnol and Kuperberg~\cite{BK18} to a \textit{$k$-token-game}, for a given $k\in \Nat$, in the scope of resolving the \HDness problem of B\"uchi automata.

In the $k$-\textit{token game}, known as $G_k$, the players proceed as in the letter game, except that now Adam has $k$ tokens that he must move after Eve has made her move, thus building $k$ runs. For Adam to win, at least one of these must be better than Eve's run. In the Boolean setting, this run must be accepting, thus witnessing that the word is in the language of the automaton. Intuitively, the more tokens Adam has, the less information he is giving Eve about the future of the word he is building.

We generalise token games to the quantitative setting, defining that the maximal value produced by Adam's runs witnesses a lower bound on the value of the word, and Eve's task is to match or surpass this value on her run.

In the Boolean setting, $G_2$ has the same winner as the letter game for B\"uchi~\cite[Corollary 21]{BK18} and coB\"uchi~\cite[Theorem 28]{BKLS20b} automata (the case of parity and more powerful automata is open). Since $G_2$ is solvable in polynomial time for B\"uchi and coB\"uchi acceptance conditions, this gives a {\sc PTime} algorithm for deciding \HDness, which avoids the determinisation used to solve the letter game directly. In the following sections we study how different token games can be used to decide $\HDness$ for different quantitative automata.

\begin{defi}[$k$-token games]\label{def:Gk}
	Consider a $\Val$ automaton $\A=(\Sigma,Q,\iota,\delta)$. 
	A configuration of the game $G_k(\A)$ for $k\geq 1$ is a tuple $(q, p_1,\dots p_k)\in Q^{k+1}$ of states, and the initial configuration is $\iota^{k+1}$.
	In a configuration $(q_i,p_{1,i},\ldots, p_{k,i})$, the game proceeds to the next configuration $(q_{i+1},p_{1,i+1},\ldots,p_{k,i+1})$ 
	as follows.  
	\begin{itemize}
		\item Adam picks a letter $\sigma_{i}$ from $\Sigma$,
		\item Eve picks a transition $\trans{q_{i}}{\sigma_{i}:\weight_{0,{i}}}{q_{i+1}}$, and
		\item Adam picks transitions, $\trans{p_{1,i}}{\sigma_{i}:\weight_{1,{i}}}{p_{1,i+1}}, \ldots, \trans{p_{k,i}}{\sigma_{i}:\weight_{k,{i}}}{p_{k,i+1}}$.
	\end{itemize}
	In the limit, a play consists of an infinite word $w$ that is derived from the concatenation of $\sigma_0,\sigma_1,\ldots$, as well as $k+1$ infinite sequences $\pi, \pi_1,\pi_2,\dots,\pi_k$ of transitions over $w$, where $\pi$ is picked by Eve and  $\pi_1,\dots,\pi_k$ by Adam. Eve wins the play if $\Val(\pi) \geq  \max(\Val(\pi_1),\dots, \Val(\pi_k))$.
	
	On finite words, $G_k$ is defined as above, except that the winning condition is a safety condition for Eve: for all finite prefixes of a play, it must be the case that the value of Eve's run is at least the value of each of Adam's runs.
\end{defi}

Cautious strategies (\cref{def:Cautious}) immediately extend to Eve's strategies in $G_k(\A)$. Unlike in the letter game, a winning strategy in $G_k(\A)$ must not necessarily be cautious: a non-cautious move does not necessarily allow Adam to win, since Adam might not have a token available to build an optimal run on the word witnessing Eve's lack of caution.

\section{Deciding History-Determinism via One-Token Games}\label{sec:G1}
Bagnol and Kuperberg showed that the one-token game $G_1$ does not suffice to characterise \HDness for B\"uchi automata~\cite[Lemma 8]{BK18}. However, it turns out that $G_1$ does characterise \HDness for all quantitative (and Boolean) automata on finite words and some quantitative automata on infinite words.

We can then use $G_1$ to decide history-determinism of some of these automata, over which the $G_1$ game is simple to decide. 
In particular, this is the case for $\Sup$ automata on finite words and $\Reach$, $\Safety$, $\Inf$ and $\DSum$ automata on finite and infinite words.

\subsection{$G_1$ Characterises \HDness for some automata}\label{sec:G1Characterizes}

\begin{thm}\label{cl:G1PresentFocused}
	Given a nondeterministic automaton $\A$ with a present-focused value function $\Val$ on finite or infinite words,
	Eve wins $G_1(\A)$ if and only if $\A$ is \HD. Furthermore, a winning strategy for Eve in $G_1(\A)$ induces an \HD strategy with the same memory.
\end{thm}

\begin{proof}
	One direction is easy: if $\A$ is \HD, Eve can use her \HD strategy to win $G_1$ by ignoring Adam's token.
	For the other direction, assume that Eve wins $G_1$. 
	
	We consider the following family of \textit{copycat strategies} for Adam in $G_1$: a copycat strategy is one where Adam moves his token in the same way as Eve until she makes a non-cautious move $t=\trans{q}{\sigma:\weight}{q'}$ after building a run $\rho$; that is, there is some word $w$ and run $\pi'$ from $q$ on $\sigma w$, such that for every run $\pi$ on $\sigma w$ starting with $t$, we have $\Val(\rho \pi') > \Val(\rho \pi)$.
	Then the copycat strategy stops copying and directs Adam's token along the run $\pi'$ and plays the word $w$. If Eve plays a non-cautious move in $G_1$ against a copycat strategy, she loses.
	Then, if Eve wins $G_1$ with a strategy $s$, she wins in particular against all copycat strategies and therefore $s$ never makes a non-cautious move against such a strategy.
	
	Eve can then play in the letter game over $\A$ with a strategy $s'$ that moves her token as $s$ would in $G_1(\A)$ assuming Adam uses a copycat strategy. Then, $s'$ never makes a non-cautious move and is therefore a cautious strategy. Since $\Val$ is present-focused, any cautious strategy, and in particular $s'$, is winning in the letter game, so $\A$ is \HD. Note that $s'$ requires no more memory than $s$.
\end{proof}

Notice that the converse of the above does not hold, namely there are value functions $\Val$ that are not present-focused, while $G_1$ still characterises $\HD$ for all $\Val$-automata, as will be shown for $\Reach$ automata.

An immediate corollary of \cref{cl:G1PresentFocused} is that $G_1$ characterises history-determinism for all automata on finite words, as all value functions on finite words are present focused.

\begin{cor}\label{cl:FiniteWords}
	Given a nondeterministic automaton $\A$ on finite words,
	Eve wins $G_1(\A)$ if and only if $\A$ is \HD, and winning strategies in $G_1(\A)$ induce \HD strategies for $\A$ of the same memory size.
\end{cor}
\begin{proof}
	A direct consequence of \cref{cl:finite-is-present-focused,cl:G1PresentFocused}.
\end{proof}

\begin{cor}\label{cl:DSumG1}
	Given a nondeterministic $\DSum$ automaton $\A$ on finite or infinite words,
	Eve wins $G_1(\A)$ if and only if $\A$ is \HD, and winning strategies in $G_1(\A)$ induce \HD strategies for $\A$ of the same memory size.
\end{cor}
\begin{proof}
	A direct consequence of \cref{cl:finite-is-present-focused,cl:DSumPresentFocused,cl:G1PresentFocused}.
\end{proof}

Considering the memory size of the \HD strategy, notice that for $\DSum$ automata on finite and infinite words, positional strategies suffice; that is these automata are \HD if and only if they are determinisable by pruning \cite[Theorem 23]{BL21}  and \cite[Section 5]{HPR16}.

\begin{lem}\label{cl:InfPresentFocused}
	The value function $\Inf$ on infinite sequences of rational values is present-focused.
\end{lem}
\begin{proof}
	Consider a cautious strategy $s$ of Eve in the letter game on an $\Inf$ automaton $\A$, and assume toward contradiction that there exists a play in which Adam wins playing against $s$. 
	
	Let $w$ be the word generated along this play. Then, the run of $\A$ on $w$ that Eve generated along the play has some value $x<\A(w)$. Let $u$ be the shortest prefix of $w$, after which Eve chose a transition $t=\trans{q}{\sigma:x}{q'}$ with value $x$, and let $\rho$ be the corresponding prefix of the run generated by Eve.
	Clearly, every continuation of $\rho$ on the suffix of $w$ from $u$ will generate a run whose value is at most $x$, thus strictly smaller than $\A(w)$.
	
	Let $\rho'$ be the longest prefix of $\rho$, for which there is a continuation on the corresponding suffix $v$ of $w$, generating a run with value $\A(w)$. Notice that such a run prefix $\rho'$ exists, since it is bounded by above by $\rho$ and below by the empty run, whose continuation on the suffix $v$ of $w$ is an arbitrary run on $w$.
	
	Then, the move $t_0=q_0\xrightarrow{\sigma_0:\weight_0}q_0'$ of Eve, played after $\rho'$ is a non-cautious transition: for the suffix $v$ of $w$, there is a run $\pi'$ from $q$ over $\sigma v$ such that $\Inf(\rho'\pi')$ is strictly greater than the value of $\Inf(\rho'\pi)$ for any $\pi$ starting with $t_0$.
	Thus, we reached a contradiction to the cautiousness of $s$.
\end{proof}

\begin{thm}\label{cl:InfIsG1}
	Given a nondeterministic $\Inf$ (or $\Safety$) automaton $\A$ on infinite words,
	Eve wins $G_1(\A)$ if and only if $\A$ is \HD, and winning strategies in $G_1(\A)$ induce \HD strategies for $\A$ of the same memory size.
\end{thm}
\begin{proof}
	A direct consequence of \cref{cl:G1PresentFocused,cl:InfPresentFocused}.
\end{proof}

We move to $\Reach$ automata on infinite words. 

Observe that the $\Reach$ value function with respect to infinite words is not present-focused (see the automaton $\A$ in \cref{fig:Reachability}), causing also a difference between history-determinism of a $\Reach$ automaton when considered with respect to finite and infinite words (see the automaton $\B$ in \cref{fig:Reachability}). 

\begin{figure}[h]
	\centering
	\begin{tikzpicture}[->,>=stealth',shorten >=1pt,auto,node distance=2cm, semithick, initial text=, every initial by arrow/.style={|->},state/.style={circle, draw, minimum size=0.5cm}]
		
		\node  (A) {$\A$};
		\node[right of = A,below of =A, initial left, state,xshift=-1.0cm,yshift=1cm] (s0) {$s_0$};
		\node[state] (s1) [ right of=s0] {$s_1$};
		
		\path 
		(s0) edge	[loop above, out=120, in=70,looseness=5] node [right,xshift=.1cm,yshift=0.1cm]{$\Sigma$} (s0)
		(s0) edge node{$\Sigma$}  (s1)		
		(s1) edge[double]	[loop above, out=120, in=70,looseness=5] node [right,xshift=.1cm,yshift=0.1cm]{$\Sigma$} (s1)
		;
		
		\node  (B) [ right of=A,xshift=3cm, yshift=0.5cm]{$\B$};
		\node[right of = B,below of =B, initial above, state,xshift=.75cm,yshift=1.5cm] (s0) {$s_0$};
		\node[state] (s1) [ left of=s0] {$s_1$};
		\node[state] (s2) [ right of=s0] {$s_2$};
		\node[state] (s3) [below of=s0, yshift=0.25cm] {$s_3$};
		
		\path 
		(s0) edge node [above]{$\Sigma$}   (s1)		
		(s0) edge node{$\Sigma$}  (s2)		
		
		(s1) edge[double]  node[above] {$a$} (s3)
		(s2) edge[double]  node[above] {$b$} (s3)
		
		(s1) edge  [out=-90, in=-180,looseness=.75] node[right,xshift=-0.15cm,yshift=0.15cm] {$b$}(s3)
		(s2) edge  [out=-90, in=0,looseness=.75] node[left,xshift=0.15cm,yshift=0.15cm] {$a$}(s3)
		
		(s3) edge[double]	[loop above, out=120, in=70,looseness=5] node [right,xshift=.05cm,yshift=0.2cm]{$\Sigma$} (s3)
		;

	\end{tikzpicture}
	\caption{An  automaton $\A$, demonstrating that the $\Reach$ value function on infinite words is not present-focused: the strategy of Eve that remains forever in $s_0$ is cautious, but does not win the letter game on $\A$.
		The $\Reach$ automaton $\B$ demonstrates another difference between reachability on finite and infinite words: It is \HD on infinite words, but not \HD on finite words.}
	\label{fig:Reachability}
\end{figure}
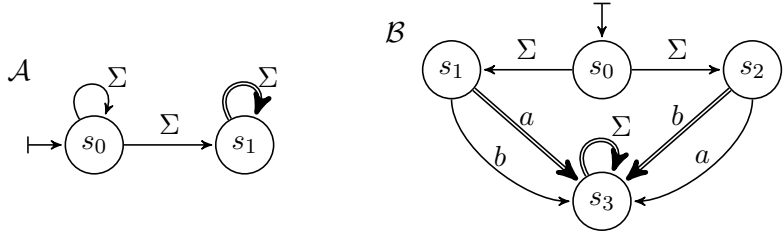

Nevertheless, there is a close connection between reachability with respect to finite and infinite words, as shown below, allowing us to show that 1-token games do characterise history-determinism also for $\Reach$ automata on infinite words.

Recall that we assume that all accepting transitions in a $\Reach$ lead to a sink state that has an accepting self loop on all the alphabet letters -- once an accepting transition has been reached, nothing else matters. Hence, $\Reach$ automata with acceptance on transitions and on states are very similar; the only difference is whether this ``heaven'' state is marked accepting or the transition leading to it is.
For simplicity, we will consider automata with acceptance on states in the proof that $G_1$ characterises history-determinism in $\Reach$ automata. 

Define a state $q$ of a $\Reach$ automaton $\A$ ``almost accepting'' if there exists a strategy $s$ in the letter game on $\A^q$ on infinite words, such that for every infinite word $w$, the run that $s$ entails on $w$ is accepting. (Notice that every accepting state is also almost accepting.)
Given a $\Reach$ automaton $\A$, we define $\Polish(\A)$ to be the $\Reach$ automaton that is derived from $\A$, by making every almost accepting state of $\A$ accepting. 

\begin{prop}\label{cl:PolishingComplexity}
	Given a $\Reach$ automaton $\A$ over an alphabet $\Sigma$ with states $Q$ and transitions $\delta$, computing $\Polish(\A)$ is in $O(|\delta|\cdot |\Sigma|)$, and the corresponding strategies that witness almost-acceptance are positional with respect to the product of $Q$ and $\Sigma$.
\end{prop}
\begin{proof}
	The ``almost acceptance game'', used to find the almost accepting states, is an adaptation of the letter game: Eve wins a play if her run on the word generated by Adam reaches an accepting state. Thus, it is a reachability game over an arena that is the product of the alphabet $\Sigma$ (for Adam's moves that choose the next letter) and $\A$ (for Eve's moves that choose her next transition). As computing the winning region of reachability games is linear in the number of the arena's transitions, and winning strategies in these games are positional, the claim directly follows.
\end{proof}

\begin{lem}\label{cl:Polished}
	Consider a $\Reach$ automaton $\A$. Then: i) If Eve wins $G_1(\A)$ on infinite words then Eve wins $G_1(\Polish(\A))$ on finite words; and ii) If $\Polish(\A)$ on finite words is \HD then $\A$ on infinite words is \HD.
\end{lem}
\begin{proof}
	Let $\A$ be the automaton on infinite words and $\A'$ stand for $\Polish(\A)$ on finite words.
	\begin{itemize}
		\item[i)] Let $s$ be a winning strategy for Eve in $G_1(\A)$. We define a strategy $s'$ for Eve in $G_1(\A')$ and show that it is winning: The strategy $s'$ follows $s$ until it reaches an almost accepting state. 
		
		If $s'$ eventually reaches an almost accepting state then Eve obviously wins, reaching an accepting state of $\A'$ with her strategy witnessing almost acceptance.
		
		Otherwise, we are guaranteed that Adam's run also never reaches an accepting state of $\A'$ (which is an almost accepting state of $\A$): Assume toward contradiction that Adam does reach such an almost accepting state $q_A$ when Eve is at a state $q_E$ that is not almost accepting. 
		Then, we claim that Adam can win $G_1(\A)$, by generating some infinite suffix $w$, over which he can reach an accepting state of $\A$, using the strategy $s_A$ that witnesses the almost acceptance of $q_A$, while Eve's strategy $s$ cannot. 
		
		Indeed, assume toward contradiction that for every word $w$, the strategy $s$ can reach an accepting state, using the knowledge that Adam starts in $q_A$ and follows $s_A$. Then, we can define a strategy $s_E$ for Eve in the letter game on $\A^{q_E}$ that also reaches an accepting state on every word $w$: It follows $s$, providing it in every step with the required knowledge about $s_A$, which is possible, since $s_A$ is fixed for the state $q_A$. This however implies that $q_E$ is almost accepting, leading to a contradiction.	
		
		\item[ii)]  Let $s'$ be Eve's winning strategy in the letter game on $\A'$. We define a strategy $s$ for Eve in the letter game on $\A$ and show that it is winning: In every play, the strategy $s$ starts just like $s'$, and continues according to the following two disjoint cases:
		\begin{itemize}
			\item If Adam generates a prefix $u$, such that $u\in L(\A')$ then, since $s'$ is winning in the letter game on $\A'$, Eve, following $s'$, is guaranteed to reach a state $q$ that is almost accepting in $\A$. 
			Then, after the prefix $u$, the strategy $s$ continues like the strategy that witnesses the almost acceptance of $q$, which guarantees to reach an accepting state in $\A$ for whatever infinite suffix that Adam generates, making Eve win.
			\item Otherwise, $s$ continues forever like $s'$: the word that Adam generates is rejecting and Eve wins.\qedhere
		\end{itemize}
		
	\end{itemize} 
\end{proof}

\begin{thm}\label{cl:Reachability}
	Given a nondeterministic $\Reach$ automaton $\A$ on infinite words,
	Eve wins $G_1(\A)$ if and only if $\A$ is \HD, and winning strategies in $G_1(\A)$ induce \HD strategies for $\A$ of the same memory size.
\end{thm}
\begin{proof}
	If  $\A$ is \HD then Eve obviously wins $G_1(\A)$, by using the same strategy as in the letter game, ignoring Adam's token.
	
	If  Eve wins $G_1(\A)$ then by \cref{cl:Polished}.i she also wins $G_1(\Polish(\A))$ on finite words, thus by \cref{cl:FiniteWords} $\Polish(\A)$ on finite words is \HD, implying by \cref{cl:Polished}.ii that $\A$ is \HD.
	
	Regarding the memory of the \HD strategy, observe that Eve's strategy $s$ in $G_1(\Polish(\A))$ is the same as her strategy in $G_1(\A)$ (see the proof of \cref{cl:Polished}); her strategy $s'$ in the letter game on $\Polish(\A)$ needs memory of the same size as $s$ (\cref{cl:FiniteWords}), and her strategy in the letter game on $\A$ either follows $s'$ or diverts to a strategy that witnesses almost-acceptance, which is positional in the arena of the letter game (\cref{cl:PolishingComplexity}).
\end{proof}

Considering the memory size of the \HD strategy, notice that for reachability (and safety) automata, positional strategies suffice; that is these automata are \HD if and only if they are determinisable by pruning \cite[Theorem 17]{BKS17}.

So far, we have shown that the $1$-token game characterises history-determinism for various quantitative automata, and in particular for $\Reach$ and $\Inf$ on infinite words. 
The $\Sup$ value function can be seen both as a generalisation of $\Reach$ to more than two weights, and as a dual of $\Inf$. However, it turns out that $\Sup$ automata behave rather differently, as demonstrated in~\cref{fig:SupNotG1}.

\begin{prop}\label{cl:SupNotG1}
	$G_1$ does not characterise history-determinism for $\Sup$ automata on infinite words.
\end{prop}
\begin{proof}
	The automaton $\A$ depicted in \cref{fig:SupNotG1} is not \HD, while Eve wins $G_1$ on it.
\end{proof}

\begin{figure}[h]
	\centering
	\begin{tikzpicture}[->,>=stealth',shorten >=1pt,auto,node distance=2cm, semithick, initial text=, every initial by arrow/.style={|->},state/.style={circle, draw, minimum size=0.5cm}]
		
		\node  (A) {$\A$};
		\node[right of = A,below of =A, initial left, state,xshift=-0.75cm,yshift=1cm] (s0) {$s_0$};
		\node[state] (s1) [ right of=s0] {$s_1$};
		
		\path 
		(s0) edge	[loop above, out=120, in=70,looseness=5] 
		node [left,xshift=-.0cm,yshift=0.1cm]{$\LW{a}{0}$} 
		node [right,xshift=.1cm,yshift=0.1cm]{$\LW{b}{3}$} 
		(s0)
		(s0) edge node[above]{$\LW{a}{0}$} node[below]{$\LW{b}{3}$} (s1)		
		(s1) edge	[loop above, out=120, in=70,looseness=5] 
		node [left,xshift=-.0cm,yshift=0.1cm]{$\LW{a}{1}$} 
		node [right,xshift=.1cm,yshift=0.1cm]{$\LW{b}{2}$} 
		(s1)
		;

	\end{tikzpicture}
	\caption{A $\Sup$  automaton $\A$, demonstrating that $G_1$ does not characterise history-determinism for $\Sup$ automata on infinite words: 
		$\A$ is not \HD as Adam can play $a$ when Eve's run is in $s_0$ and $b$ when Eve's run is in $s_1$. If Eve stays in $s_0$, then the word has value 1 and Eve's run has value 0; if Eve goes to $s_1$, then the word has value 3 but Eve's run has value 2.
		Eve wins $G_1$ by moving to $s_1$ once Adam's token is in $s_1$. If Adam stays in $s_0$, they have the same run; if Adam moves and plays $b$ before Eve moves, she gets value 3 and wins; if he doesn't, then Eve gets the same value as Adam.
	}
	\label{fig:SupNotG1}
\end{figure}
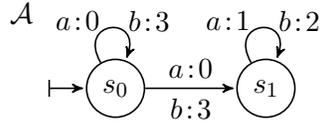

\subsection{Solving $G_1$ and Deciding \HDness}\label{sec:solveG1}
\hfill\\[-6pt]

\noindent
We continue with solving the \HDness problem of some of the automata types discussed in \cref{sec:G1Characterizes}, for which it is easy to solve the $G_1$ game.

We start with $\Reach$ and $\Safety$ automata on finite words, for which solving $G_1$ reduces to solving a safety game.

\begin{thm}\label{cl:solveG1-reach-safe}
	Deciding whether a $\Reach$ or $\Safety$ automaton on finite words is \HD can be done in time cubic in the size of $\A$.
\end{thm}

\begin{proof}
	By \cref{cl:FiniteWords}, it is enough to solve $G_1(\A)$.
	
	Given a $\Safety$ or $\Reach$ automaton $\A=(\Sigma,Q,\iota,\delta)$ on finite words,  $G_1(\A)$ reduces to a safety game, whose positions $(\sigma,q,q', t)\in \Sigma\cup\{\varepsilon\}\times Q^2\times \{L, E, A\}$ consist of a possibly empty letter $\sigma$ representing the last letter played, a pair of states $(q,q')$, one for Eve and one for Adam, which keep track of the end of the current run built by each player, and a turn variable $t\in \{L,E,A\}$ indicating whether it is Adam's turn to give a letter ($L$), Eve's turn to choose a transition ($E$), or Adam's turn to choose a transition ($A$). The initial position is $(\epsilon, \iota,\iota,L)$. The moves and position ownership encode the permitted moves in $G_1(\A)$.
	
	In both cases $G_1(\A)$ is a safety game: in the $\Reach$ case, Eve wins if the play remains in positions where either neither player's token has reached the target, or she has reached the target; in the $\Safety$ cases, Eve wins if the play remains in positions where either her token is in the safe region, or Adam's token has left the safe region. In both cases, the game can be represented by a cubic-sized arena, which can then be solved in linear time by computing Adam's attractor to the non-safe region.
\end{proof}

While $G_1(\A)$ is a safety game if $\A$ operates on finite words, this is no longer the case for $\A$ on infinite words, as, for example in the $\Reach$ case, Eve could reach the target any time after Adam, turning a potentially losing play prefix into a winning one.
Nevertheless, solving $G_1$ is still simple, reducing it to solving weak games.

\begin{thm}\label{cl:solveG1-reach-safe-infinite-words}
	Deciding whether a $\Reach$ or $\Safety$ automaton $\A$ on  infinite words is \HD can be done in time cubic  in the size of $\A$.
\end{thm}

\begin{proof}
	By \cref{cl:Reachability,cl:InfIsG1}, it is enough to solve $G_1(\A)$.
	
	We can encode $G_1(\A)$ as a weak game on an arena that consists, as in the proof of~\cref{cl:solveG1-reach-safe}, of the product of two copies of $\A$ to keep track of each player's token, the alphabet to indicate Adam's last choice of letter, and a variable to indicate whose turn it is to play. For the $\Safety$ case, a move is good for Eve, encoded with priority $0$, if after the move either Adam's token is out of the safe region, or her token is within the safe region. Other moves are bad for Eve, encoded with priority $1$. Eve wins plays that eventually remain within the good region.
	Similarly, in the $\Reach$ case, a move is good for Eve, encoded with priority $0$ if either her token has reached the target or Adam's token has not reached the target. In both cases, the winning condition is a weak condition since there are no plays alternating good and bad moves infinitely often.
	In both cases, the resulting game can then be solved in linear time~\cite{HMS06}.
	
	Note that alternatively, for a $\Safety$ automaton $\A$, one can encode $G_1(\A)$ as a safety game. Indeed, positions in which Eve's token has left the safe region while Adam's token is in a state with a non-empty language are winning for Adam, and can be marked unsafe for Eve, while positions in which Adam's token is in a state with empty language are winning for Eve, and can be marked safe for Eve.  Then Eve wins if and only if the play remains in safe positions, that is, either her token is in the safe region of $\A$ or Adam's token is in a state with empty language.
	
	Analogously, there is an alternative solution also for a $\Reach$ automaton $\A$, whereby $G_1(\A)$ is encoded as a reachability game, in which the target positions are the almost accepting ones (which, by \cref{cl:PolishingComplexity}, can be computed in time cubic in the size of $\A$): Positions in which Eve's token is in an almost accepting state are the ones she needs to reach.
\end{proof}

Solving $G_1$ for $\Sup$ automata reduces, as in some of the previous cases, to solving safety games.

\begin{thm}\label{cl:FiniteSupPtime}
	Deciding whether a $\Sup$ automaton on finite words is \HD is in \PTime, namely in $O(|\Sigma|n^2k)$, where $\Sigma$ is the automaton's alphabet, $k$ the number of weights and $n$ the number of  states.
\end{thm}

\begin{proof}
	By \cref{cl:FiniteWords}, it is enough to solve $G_1$.
	
	Given a $\Sup$ automaton $\A=(\Sigma,Q,\iota,\delta)$ with weights $W$,  $G_1(\A)$ reduces to a safety game, by taking, as in the previous proofs, the product of two copies of $\A$, $\Sigma$, and a variable to keep track of whose turn it is to play. In addition, we use an additional variable $x_E$ to keep track of the maximal weight on Eve's run so far. 
	
	The winning condition for Eve is a safety condition: Adam wins if he picks a move with a weight higher than $x_E$, the maximal weight on Eve's run. Then plays in this game are in bijection with plays of $G_1(\A)$, and Eve wins if and only if she can avoid Adam choosing a transition with a larger weight than $x_E$, that is, if she can win $G_1(\A)$. 
	
	Then, solving $G_1(\A)$ reduces to solving this safety game, which can be done in time linear in the number of positions of the arena, which is $3|\Sigma|n^2k$.
\end{proof}

\NotNeeded{
	\begin{thm}\label{cl:InfPtime}
		Deciding whether an $\Inf$ automaton on finite words is \HD is in \PTime, namely in $O(|\Sigma|n^2k^2)$, where $\Sigma$ is the automaton's alphabet, $k$ the number of weights and $n$ the number of  states.
	\end{thm}
	\begin{proof}
		By \cref{cl:FiniteWords}, it is enough to solve $G_1$.
		
		Given an $\Inf$ automaton $\A$ on finite words, $G_1(\A)$ can be reduced to a safety game that is similar to the safety game detailed in the proof of \cref{cl:FiniteSupPtime}, except that instead of keeping Eve's maximal value along her run, we need to keep the minimal value along Adam's run in some variable $x_A$, and the safety condition for Eve is that her current value must always be at least as big as $x_A$ and Adam's next move. Since Adam plays after Eve in each round of the game, we also need to keep Eve's last value, thus having $3|\Sigma|n^2k^2$ positions.
	\end{proof}
}

Solving $G_1$ for $\Inf$ automata reduces to solving safety games, when the automata operate on finite words, and to solving weak games, when the automata operate on infinite words.

\begin{thm}\label{cl:InfPtimeInfiniteWords}
	Deciding whether an $\Inf$ automaton on finite or infinite words is \HD is in \PTime, namely in $O(|\Sigma|n^2k^2)$, where $\Sigma$ is the automaton's alphabet, $k$ the number of weights and $n$ the number of  states.
\end{thm}
\begin{proof}
	By \cref{cl:InfIsG1} and \cref{cl:FiniteWords}, it is enough to solve $G_1$ on the $\Inf$ automaton $\A$.
	
	Analogously to the previous proofs, we can encode $G_1(\A)$ as a product of two copies of $\A$ to keep track of each player's token, the alphabet to keep track of the last letter chosen by Adam, a variable to keep track of whose turn it is to play, and a pair of variables to remember the least value read so far by each player's token. (Notice that as opposed to the $\Sup$ case, proved in \cref{cl:FiniteSupPtime}, the encoding in this case needs to keep the least value read so far by both players, and cannot do with only one of them.)
	
	In the case of finite words, the winning condition for Eve is a safety condition: the least value seen so far by Eve's token must always be at least as high as the one seen by Adam's token.
	
	In the case of infinite words, the winning condition for Eve is that eventually the value seen by Adam's token must remain at least as high as the one seen by Eve's token. Since there are only finitely many possible values and the least value for each token can only decrease along the run, this is a weak condition.
	
	In both cases the resulting game, of size $O(|\Sigma|n^2k^2)$, can be solved in linear time~\cite{HMS06}.
\end{proof}

Next, we show that solving $G_1$ is in {\sc NP}$\cap${\sc co-NP} for $\DSum$ automata.

\begin{thm}\label{cl:DSum-NPcoNP}
	For every $\lambda\in\Rat\cap(0,1)$, deciding whether a $\lambda$-$\DSum$ automaton $\A$, on finite or infinite words, is \HD is in {\sc NP}$\cap${\sc co-NP}\footnote{It was already known for finite words~\cite{FLW20}. It is perhaps surprising for infinite words, given the NP-hardness result in~\cite[Theorem 6]{HPR16}. In consultation with the authors, we have confirmed that there is an error in the hardness proof.}. 
\end{thm}

\begin{proof}
	Consider a $\lambda$-DSum automaton $\A=(\Sigma,Q,\iota,\delta)$, where the weight of a transition $t$ is denoted by $\gamma(t)$.
	By \cref{cl:DSumG1}, it suffices to show that solving $G_1(\A)$ is {\sc NP}$\cap${\sc co-NP}. We achieve this by reducing solving $G_1(\A)$ to solving a discounted-sum threshold game, which Eve wins if the $\DSum$ of a play is non-negative. It is enough to consider infinite games, as they also encode finite games, by allowing Adam to move to a forever-zero-position in each of his turns.
	
	The reduction follows the same pattern as that in the proof of \cref{cl:FiniteSupPtime}: we represent the arena of the game $G_1(\A)$ as a finite arena, and encode its winning condition, which requires the difference between the $\DSum$ of two runs to be non-negative, as a threshold $\DSum$ winning condition. Note first that the difference between the $\lambda$-$\DSum$ of the two sequences $\weight_0 \weight_1...$ and $\weight'_0 \weight'_1...$ of weights is equal to the $\lambda$-$\DSum$ of the sequence of differences $d_0=(\weight_0-\weight'_0), d_1=(\weight_1-\weight'_1), \ldots$, as follows:
	$
	(\sum_{i=0}^\infty \lambda^i \weight_i ) - \sum_{i=0}^\infty \lambda^i \weight'_i  = \sum_{i=0}^\infty \lambda^i (\weight_i-\weight'_i) 
	$.
	
	We now describe the $\DSum$ arena $G$ in which Eve wins with a non-strict $0$-threshold objective if and only if she wins $G_1(\A)$. (See an example in \cref{fig:DsumExample}.)
	The arena has positions in $(m,\sigma,t,q,q')\in \{L,E,A\}\times \Sigma\cup \{\varepsilon\} \times \delta\cup \{\varepsilon\} \times Q^2$ where $m$ denotes the move type, having $L$ for Adam choosing a letter, $E$ for Eve choosing a transition and $A$ for Adam choosing a transition; $\sigma$ is the last played letter if $m=E$ or $m=A$ and $\varepsilon$ otherwise; $t$ is the transition just played by Eve if $m=A$ and $\varepsilon$ otherwise; and the states $q,q'$ represent the positions of Eve and Adam's tokens. 
	
	A move of Adam that chooses a transition $t'=\trans{q'}{\sigma:\weight}{q''}$, namely a move $(A,\sigma,t,q,q')\rightarrow(L,\varepsilon,\varepsilon,q,q'')$, is given weight $\gamma(t)-\gamma(t')$, that is, the difference between the weights of the transitions chosen by both players.
	Other transitions are given weight $0$. 
	Observe that we need to compensate for the fact that only one edge in three is weighted. One option to do it is to take a discount factor $\lambda'= \lambda^{\frac{1}{3}}$ for the $\DSum$ game $G$. Yet, $\lambda'$ can then be irrational, which somewhat complicates things. Another option is to consider discounted-sum games with multiple discount factors \cite{And06} and choose three rational discount factors $\lambda', \lambda'', \lambda''' \in \Rat \cap (0,1)$, such that $\lambda' \cdot \lambda'' \cdot \lambda'''= \lambda$. Since the first two weights in every triple are $0$, only the multiplication of the three discount factors toward the third weight is what matters. For $\lambda=\frac{p}{q}$, where $p<q$ are positive integers, one can choose $\lambda'=\frac{4p}{4p+1}, \lambda''=\frac{4p+1}{4p+2}$, and $\lambda'''=\frac{2p+1}{2q}$.
	
	Plays in $G_1(\A)$ and in $G$ are in bijection. It now suffices to argue that the winning condition of $G$, namely that the $(\lambda', \lambda'', \lambda''')$-$\DSum$ of the play is non-negative, correctly encodes the winning condition of $G_1(\A)$, meaning that the difference between the $\lambda$-$\DSum$ of Eve's run and of Adam's run is non-negative. 
	
	Let $d_0 d_1\ldots$ be the sequence of weight differences between the transitions played by both players in $G_1(\A)$, and let $\lambda_0, \lambda_1, \ldots$ and $w_0, w_1, \ldots$ be the corresponding sequences of discount factors and weights in the  $(\lambda', \lambda'', \lambda''')$-$\DSum$ game, respectively, where for every $i=(0 \mod 3)$, we have $w_i=0$ and $\lambda_i=\lambda'$, for every $i=(1 \mod 3)$, we have $w_i=0$ and $\lambda_i=\lambda''$, and for every $i=(2 \mod 3)$, we have $w_i=d_i$ and $\lambda_i=\lambda'''$. Then the value of the $(\lambda', \lambda'', \lambda''')$-$\DSum$ sequence is equal to the required $\DSum$ sequence multiplied by $\lambda' \cdot \lambda''$:
	
	\[
	(\lambda', \lambda'', \lambda''')\text{-}\DSum = 
	\sum_{i=0}^\infty (0 \cdot \prod_{j=0}^{3i-1} \lambda_j ~+ 0 \cdot \prod_{j=0}^{3i} \lambda_j ~+ w_{3i+2} \cdot \prod_{j=0}^{3i+1} \lambda_j)
	= \lambda' \cdot \lambda'' \cdot \sum_{i=0}^\infty \lambda^i d_i 
	\]
	
	Hence Eve wins the game $G_1(\A)$ if and only if she wins the $0$-threshold $(\lambda', \lambda'', \lambda''')$-$\DSum$ game over $G$.
	As $G$ has a state-space polynomial in the state-space of $\A$ and solving $\DSum$-games is in {\sc NP}$\cap${\sc coNP}~\cite{And06}, solving $G_1(\A)$, and therefore deciding whether $\A$ is \HD, is also in {\sc NP}$\cap${\sc coNP}.
\end{proof}

\newcommand{\NodeL}[3]{\node[rectangle,draw] (#1) [#2]{$\begin{array}{c}\!\!\! L, \epsilon, \epsilon\!\!\!\\ \!\!#3\!\!\end{array}$};}

\newcommand{\NodeE}[4]{\node[ellipse,draw,inner sep=0,outer sep=0] (#1) [#2]{$\begin{array}{c}\!\!\! E, #3, \epsilon\!\!\!\\ \!\!\! #4 \!\!\!\end{array}$};}

\newcommand{\NodeA}[5]{\node[rectangle,draw] (#1) [#2]{$\begin{array}{c}\!\!\!A, #3, #4\!\!\!\\ \!\!#5\!\!\end{array}$};}

\newcommand{\ArrowLE}[3]{(#1) edge node[#2] {$\lambda', 0$} (#3)}
\newcommand{\ArrowEA}[2]{(#1) edge node[yshift=-0.1cm] {$\lambda'', 0$} (#2)}
\newcommand{\ArrowAL}[4]{(#1) edge node [#2] {$\lambda''', #3$}(#4)}
\newcommand{\PreArrowAL}[4]{(#1) edge [->,out=0,in=0,looseness=0] node [#2] {$\lambda''', #3$}(#4)}

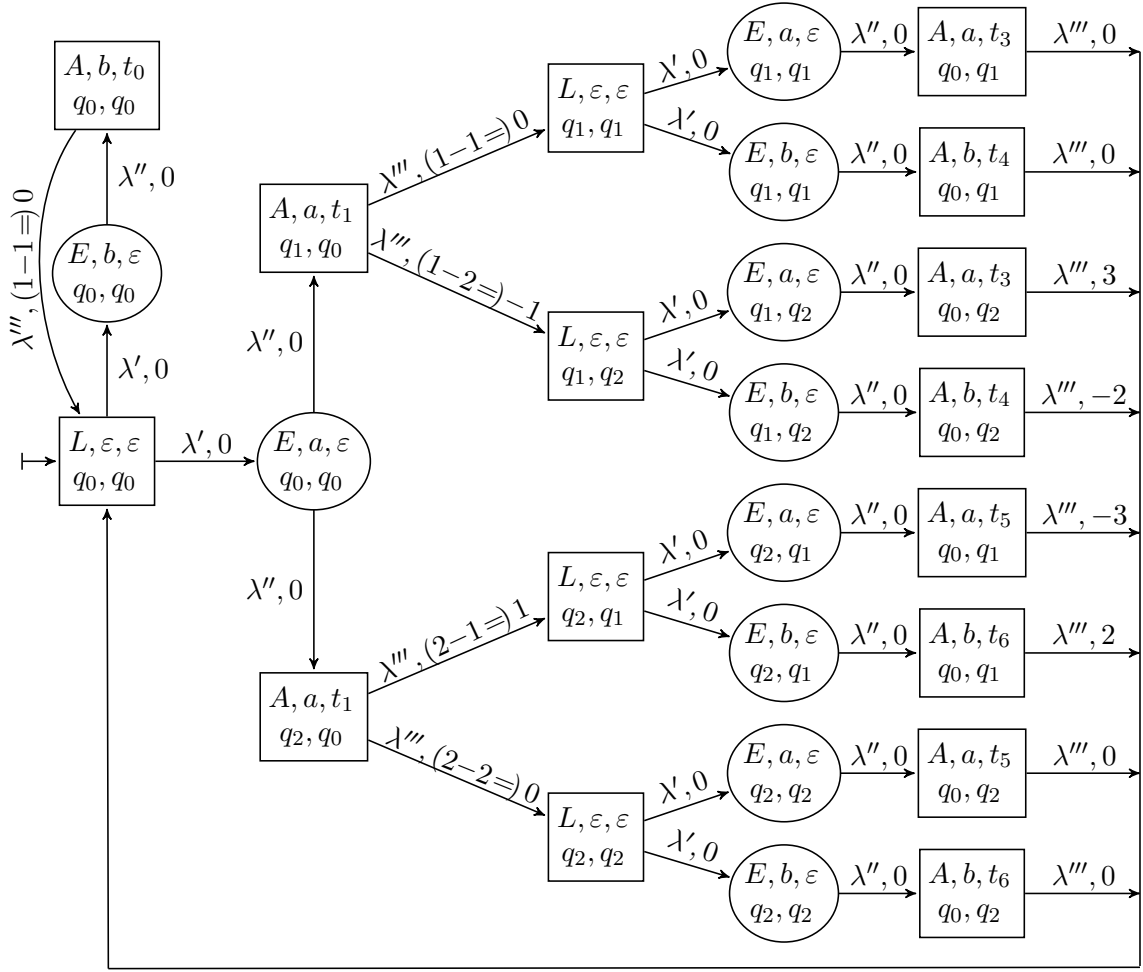
\begin{figure}[hb]
	\centering
	\begin{tikzpicture}[->,>=stealth',shorten >=1pt,auto,node distance=2cm, semithick, initial text=, every initial by arrow/.style={|->},state/.style={circle, draw, minimum size=0.5cm}]
		
		\node (A) {A $\DSum$ automaton $\A$ with a discount factor $\lambda=\frac{5}{7}$};
		\node[below of =A, initial above, state,xshift=-0.0cm,yshift=-0.0cm] (q0) {$q_0$};
		\node[state] (q1) [ right of=q0,xshift=0.75cm,yshift=0cm] {$q_1$};
		\node[state] (q2) [ left of=q0,xshift=-0.75cm,yshift=0cm] {$q_2$};
		
		\path 
		(q0) edge	[out=250, in=300,looseness=4] node [below,xshift=0.1cm,yshift=0.1cm]{$t_0\,\,\,\LW{b}{0}$} (q0)
		(q0) edge node[below,xshift=0.05cm,yshift=0.0cm] {$t_1\,\,\,\LW{a}{1}$}   (q1)		
		(q0) edge node[below,xshift=0.0cm,yshift=0.0cm] {$t_2\,\,\,\LW{a}{2}$}  (q2)		

		(q1) edge	[out=170, in=40,looseness=0.75] node [above,xshift=0cm,yshift=-0.15cm,rotate=-15]{$t_3\,\,\,\LW{a}{4}$} (q0)
		(q1) edge	[out=140, in=70,looseness=0.85] node [above,xshift=0cm,yshift=-0.1cm,rotate=-10]{$t_4\,\,\,\LW{b}{1}$} (q0)
		(q2) edge	[out=10, in=140,looseness=0.75] node [above,xshift=0.05cm,yshift=-0.15cm,rotate=12]{$t_5\,\,\,\LW{a}{1}$} (q0)
		(q2) edge	[out=40, in=110,looseness=0.85] node [above,xshift=0cm,yshift=-0.1cm,rotate=10]{$t_6\,\,\,\LW{b}{3}$} (q0)
		;
		
		\node[below of =A, xshift=0.0cm,yshift=-2.0cm] (G) 
		{A $\DSum$ game $G$ with discount factors $\lambda'=\frac{20}{21}, \lambda''=\frac{21}{22}$, and $\lambda'''=\frac{11}{14}$};

		\NodeL{L0}{below of =G,xshift=-5cm,yshift=-4.5cm,initial left}{q_0,q_0}

		\NodeE{E01}{right of =L0, xshift=0.75cm}{a}{q_0,q_0}
		\NodeE{E02}{above of =L0, yshift=0.5cm}{b}{q_0,q_0}

		\NodeA{A01}{above of =E01, xshift=-0.0cm,yshift=1.1cm}{a}{t_1}{q_1,q_0}
		\NodeA{A02}{below of =A01, yshift=-4.5cm}{a}{t_1}{q_2,q_0}
		\NodeA{A03}{above of =E02,yshift=0.5cm}{b}{t_0}{q_0,q_0}

		\NodeL{L1}{right of =A01, xshift=1.75cm,yshift=1.6cm}{q_1,q_1}
		\NodeL{L2}{below of =L1,yshift=-1.3cm}{q_1,q_2}
		\NodeL{L3}{below of =L2,yshift=-1.2cm}{q_2,q_1}
		\NodeL{L4}{below of =L3,yshift=-1.2cm}{q_2,q_2}

		\NodeE{E1}{right of =L1,xshift=0.5cm,yshift=0.75cm}{a}{q_1,q_1}
		\NodeE{E2}{below of =E1,yshift=0.4cm}{b}{q_1,q_1}
		\NodeE{E3}{below of =E2,yshift=0.4cm}{a}{q_1,q_2}
		\NodeE{E4}{below of =E3,yshift=0.4cm}{b}{q_1,q_2}
		\NodeE{E5}{below of =E4,yshift=0.4cm}{a}{q_2,q_1}
		\NodeE{E6}{below of =E5,yshift=0.4cm}{b}{q_2,q_1}
		\NodeE{E7}{below of =E6,yshift=0.4cm}{a}{q_2,q_2}
		\NodeE{E8}{below of =E7,yshift=0.4cm}{b}{q_2,q_2}
		
		\NodeA{A1}{right of =E1, xshift=0.5cm}{a}{t_3}{q_0,q_1}
		\NodeA{A2}{right of =E2, xshift=0.5cm}{b}{t_4}{q_0,q_1}
		\NodeA{A3}{right of =E3, xshift=0.5cm}{a}{t_3}{q_0,q_2}
		\NodeA{A4}{right of =E4, xshift=0.5cm}{b}{t_4}{q_0,q_2}
		\NodeA{A5}{right of =E5, xshift=0.5cm}{a}{t_5}{q_0,q_1}
		\NodeA{A6}{right of =E6, xshift=0.5cm}{b}{t_6}{q_0,q_1}
		\NodeA{A7}{right of =E7, xshift=0.5cm}{a}{t_5}{q_0,q_2}
		\NodeA{A8}{right of =E8, xshift=0.5cm}{b}{t_6}{q_0,q_2}
		
		\node (LL1)[right of =A1, xshift=0.25cm, inner sep=0,outer sep=0]{};
		\node (LL2)[right of =A2, xshift=0.25cm, inner sep=0,outer sep=0]{};
		\node (LL3)[right of =A3, xshift=0.25cm, inner sep=0,outer sep=0]{};
		\node (LL4)[right of =A4, xshift=0.25cm, inner sep=0,outer sep=0]{};
		\node (LL5)[right of =A5, xshift=0.25cm, inner sep=0,outer sep=0]{};
		\node (LL6)[right of =A6, xshift=0.25cm, inner sep=0,outer sep=0]{};
		\node (LL7)[right of =A7, xshift=0.25cm, inner sep=0,outer sep=0]{};
		\node (LL8)[right of =A8, xshift=0.25cm, inner sep=0,outer sep=0]{};
		
		\node (LL9)[right of =A8, yshift=-1cm,xshift=0.25cm, inner sep=0,outer sep=0]{};
		\node (LL10)[left of =LL9, yshift=0cm,xshift=-11.75cm, inner sep=0,outer sep=0]{};

		\path

		\ArrowLE{L0}{xshift=0.0cm,yshift=-0.1cm}{E01}
		\ArrowLE{L0}{right,xshift=0.0cm,yshift=-0.0cm}{E02}

		(E01) edge node[left,yshift=-0.0cm] {$\lambda'', 0$} (A01)
		(E01) edge node[left,yshift=-0.0cm] {$\lambda'', 0$} (A02)
		(E02) edge node[right,yshift=-0.0cm] {$\lambda'', 0$} (A03)

		\ArrowAL{A01}{xshift=1.2cm,yshift=0.4cm,rotate=25}{(1{-}1\!=\!)\,0}{L1}
		\ArrowAL{A01}{xshift=-1.4cm,yshift=0.5cm,rotate=-25}{(1{-}2\!=\!)\!-\!1}{L2}
		\ArrowAL{A02}{xshift=1.2cm,yshift=0.4cm,rotate=25}{(2{-}1\!=\!)\,1}{L3}
		\ArrowAL{A02}{xshift=-1.2cm,yshift=0.3cm,rotate=-22}{(2{-}2\!=\!)\,0}{L4}

	(A03) edge[out=235,in=120,looseness=0.85] node [xshift=-0.2cm,yshift=-1.2cm,rotate=90] {$\lambda''', (1{-}1\!=\!)\,0$}(L0)

		\ArrowLE{L1}{xshift=0.5cm,yshift=0.05cm,rotate=20}{E1}
		\ArrowLE{L1}{xshift=-0.5cm,yshift=0.05cm,rotate=-20}{E2}
		\ArrowLE{L2}{xshift=0.5cm,yshift=0.1cm,rotate=20}{E3}
		\ArrowLE{L2}{xshift=-0.5cm,yshift=0.1cm,rotate=-20}{E4}
		\ArrowLE{L3}{xshift=0.5cm,yshift=0.1cm,rotate=20}{E5}
		\ArrowLE{L3}{xshift=-0.5cm,yshift=0.1cm,rotate=-20}{E6}
		\ArrowLE{L4}{xshift=0.5cm,yshift=0.1cm,rotate=20}{E7}
		\ArrowLE{L4}{xshift=-0.5cm,yshift=0.1cm,rotate=-20}{E8}
		 
		\ArrowEA{E1}{A1}
		\ArrowEA{E2}{A2}
		\ArrowEA{E3}{A3}
		\ArrowEA{E4}{A4}
		\ArrowEA{E5}{A5}
		\ArrowEA{E6}{A6}
		\ArrowEA{E7}{A7}
		\ArrowEA{E8}{A8}
		
		\PreArrowAL{A1}{xshift=0.0cm,yshift=-0.1cm}{0}{LL1}
		\PreArrowAL{A2}{xshift=0.0cm,yshift=-0.1cm}{0}{LL2}
		\PreArrowAL{A3}{xshift=0.0cm,yshift=-0.1cm}{3}{LL3}
		\PreArrowAL{A4}{xshift=0.0cm,yshift=-0.1cm}{-2}{LL4}
		\PreArrowAL{A5}{xshift=0.0cm,yshift=-0.1cm}{-3}{LL5}
		\PreArrowAL{A6}{xshift=0.0cm,yshift=-0.1cm}{2}{LL6}
		\PreArrowAL{A7}{xshift=0.0cm,yshift=-0.1cm}{0}{LL7}
		\PreArrowAL{A8}{xshift=0.0cm,yshift=-0.1cm}{0}{LL8}
		
		(LL1) edge [-,semithick,out=180,in=180,looseness=0] (LL9) 
		(LL9)edge [-,semithick,out=135,in=180,looseness=0] (LL10) 
		(LL10)edge [semithick,out=0,in=270,looseness=0] (L0)
		
		;

	\end{tikzpicture}
	\caption{An example of a $\DSum$ automaton $\A$ and the corresponding $\DSum$ game $G$ with multiple discount factors, as per the proof of \cref{cl:DSum-NPcoNP}, such that Eve wins $G1(\A)$ if and only if she wins $G$ with respect to the non-strict $0$-threshold. Rectangular positions in $G$ are of Adam and ellipses of Eve.}
	\label{fig:DsumExample}
\end{figure}

$\DSum$ games are positionally determined~\cite{Sha53,ZP95,And06} so this algorithm also computes a finite-memory witness of \HDness for $\A$ that is of polynomial size in the state-space of $\A$. However, a positional witness also exists~\cite[Section 5]{HPR16}.

One might be tempted to think that the proof of \cref{cl:DSum-NPcoNP} generalises to discounted-sum automata with multiple discount factors \cite{BH21}. Yet, this is not the case: While the value function of discounted-summation with multiple discount factors is indeed present-focused, using an argument analogous to the one in \cite[Lemma 22]{BL21}, and the proof of \cref{cl:DSum-NPcoNP} indeed uses discounted-sum games with multiple discount factors, the issue is that the value difference between the discounted-summation of two weight sequences is much more involved with multiple discount factors than with a single one. Hence, the representation that we use in the proof of \cref{cl:DSum-NPcoNP} to capture this value difference no longer holds in the case of multiple discount factors. We leave open the question of the complexity (and decidability) of \HDness of discounted-sum automata with multiple discount factors.

\section{Deciding History-Determinism via Two Token Games}\label{sec:G2}

In this section we solve the \HDness problem of  $\Sup$, $\LimSup$ and $\LimInf$ automata on infinite words 
via two-token games. As is the case with B\"uchi and coB\"uchi automata \cite[Lemma 8]{BK18}\footnote{The B\"uchi automaton presented in \cite[Lemma 8]{BK18} is a weak one, so it can also be viewed as a coB\"uchi one.}, one-token games do not characterise \HDness of $\LimSup$ and $\LimInf$ automata. We showed in~\cref{sec:G1} that this is also the case for $\Sup$ automata on infinite words (unlike on finite words where one token does suffice).

For $\Sup$ and $\LimInf$, a possible  approach is to solve the letter game directly: an equivalent deterministic automaton can track the value of a word, and the winning condition of the letter game corresponds to comparing Eve's run to the one of the deterministic automaton. Unfortunately, determinising both $\Sup$ and $\LimInf$ automata is exponential in the number of states \cite[Theorem 13]{CDH10}, so the new game is large. In addition, for $\LimInf$ automata, its winning condition, which compares the $\LimInf$ value of two runs, is non-standard and needs additional work to be encoded into a parity game.
For $\LimSup$ automata the situation is even worse, as they are not necessarily equivalent to deterministic $\LimSup$ automata, so it is not obvious whether the winner of the letter game is decidable at all.

Here we first show, in \cref{sec:G2Characterizes}, that the $2$-token-game approach, used to resolve \HDness of B\"uchi and coB\"uchi automata, can be generalised to $\Sup$, $\LimSup$ and $\LimInf$ automata. 
Our proofs for the $\LimSup$ and $\LimInf$ cases follow the same structure, while relying on the $G_2$ characterisation of \HDness for B\"uchi and coB\"uchi automata respectively. We then show that $G_2$ also characterises the \HDness of $\Sup$ automata on infinite words, via reduction to the $\LimSup$ case. 

We then analyse the complexity of solving $G_2$ for these three automata classes, in \cref{sec:solveG2}. Perhaps surprisingly (since the naive approach to solving the letter game seems harder for $\LimSup$), we show that $G_2$ is solvable in quasipolynomial time for $\LimSup$ while for $\LimInf$ our algorithm is exponential in the number of weights (but not in the number of states). For $\Sup$ automata on infinite words, solving $G_2$ can be done in polynomial time, as it reduces to solving a coB\"uchi game of cubic size.

Without loss of generality, we assume the weights to be $\{1, 2, \ldots\}$.

\subsection{$G_2$ Characterises \HDness for some automata}\label{sec:G2Characterizes}
\hfill\\[-6pt]

\noindent
This section is dedicated to proving that a $\LimSup$, $\LimInf$ or $\Sup$ automaton on infinite words is \HD if and only if Eve wins the $2$-token game on it. In both the $\LimSup$ and the $\LimInf$ cases, the structure of the argument is similar.
One direction is immediate: if an automaton $\A$ is \HD, then Eve can use the letter-game strategy to win in $G_2(\A)$, ignoring Adam's tokens. The other direction requires more work. We use an additional notion, that of $k$-\HDness, which generalises \HDness, in the sense that Eve maintains $k$ runs, rather than only one, and needs at least one of them to be optimal. We will then show that if Eve wins $G_2(\A)$, then $\A$ is $k$-\HD for a finite $k$ (namely, the number of weights in $\A$ minus one); this is where the argument differs slightly according to the value function. Finally, we will show that for automata that are $k$-\HD, for any finite $k$, a strategy for Eve in $G_2(\A)$ can be combined with the $k$-\HD strategy to obtain a strategy for her in the letter game. 

Many of the tools used in this proof are familiar from the $\omega$-regular setting~\cite{BK18,BKLS20b}. The main novelty in the argument is the decomposition of the $\LimSup$ ($\LimInf$) automaton $\A$ with $k$ weights into $k-1$ B\"uchi (coB\"uchi) automata $\A_2,\ldots,\A_k$, each $A_i$ of which recognises the language of words of value at least $i$, that are \HD whenever Eve wins $G_2(\A)$. (The converse does not hold, namely $\A_2,\ldots,\A_k$ can be \HD even if Eve loses $G_2(\A)$ -- see \cref{fig:THDnotHD}.) 
In both cases, the \HD strategies for $\A_2,\ldots,\A_k$ can then be combined to prove the $k$-\HDness of $\A$.

\cref{fig:Flow} illustrates the flow of our arguments.

Finally, we show that $G_2$ also characterises the \HDness of $\Sup$ automata on infinite words using the characterisation for $\LimSup$.

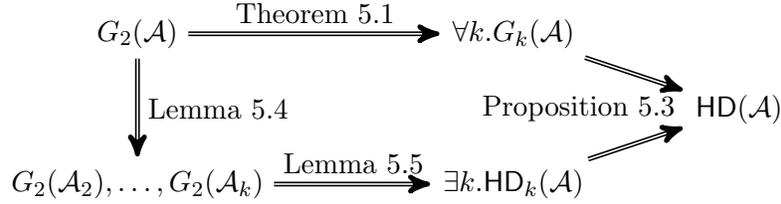
\begin{figure}[h]
	\centering
	\begin{tikzpicture}[->,>=stealth',shorten >=1pt,auto,node distance=2cm, semithick, initial text=, every initial by arrow/.style={|->},state/.style={circle, draw, minimum size=0.5cm}]
		
		\node  (G2) {$G_2(\A)$};
		\node[right of = G2,xshift=3cm] (Gk) {$\forall k. G_k(\A)$};
		\node[below of = G2] (G2x) {$G_2(\A_2),\ldots,G_2(\A_k)$};
		\node[below of = Gk] (HDk) {$\exists k. \HD_k(\A)$};
		\node[right of = Gk,xshift=1cm, yshift=-1cm] (HD) {$\HD(\A)$};
		
		\path 
		(G2) edge[double] node [above]{\cref{cl:2TokenKToken}}   (Gk)		
		(G2) edge[double] node [right]{\cref{cl:decomposeG2}}   (G2x)		
		(G2x) edge[double] node [above]{\cref{cl:G2Gk}}   (HDk)		
		(Gk) edge[double] node [below,xshift=-0.75cm, yshift=-0.15cm]{\cref{cl:kHD-G2iffHD}}   (HD)		
		(HDk) edge[double]  (HD)		
		;
	\end{tikzpicture}
	\caption{The flow of arguments for showing that $G_2(\A) \implies \HD(\A)$ for a $\LimInf$ or $\LimSup$ automaton $\A$.}
	\label{fig:Flow}
\end{figure}

We first generalise to quantitative automata Bagnol and Kuperberg's key insight that if Eve wins $G_2$, then she also wins $G_k$ for all $k$~\cite[Thm 14]{BK18}.

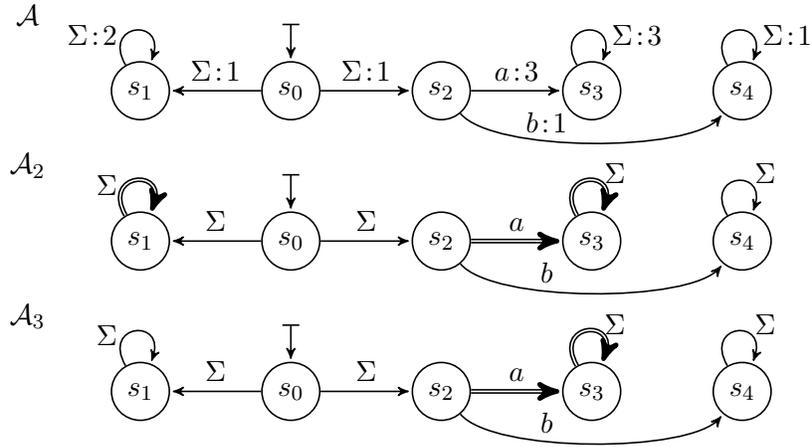
\begin{figure}[h]
	\centering
	\begin{tikzpicture}[->,>=stealth',shorten >=1pt,auto,node distance=2cm, semithick, initial text=, every initial by arrow/.style={|->},state/.style={circle, draw, minimum size=0.5cm}]
		
		\node  (A) {$\A$};
		\node[right of = A,below of =A, initial above, state,xshift=1.5cm,yshift=1cm] (s0) {$s_0$};
		\node[state] (s1) [ left of=s0] {$s_1$};
		\node[state] (s2) [ right of=s0] {$s_2$};
		\node[state] (s3) [right of=s2] {$s_3$};
		\node[state] (s4) [right of=s3] {$s_4$};
		
		\path 
		(s0) edge node [above]{$\LW{\Sigma}{1}$}   (s1)		
		(s0) edge node{$\LW{\Sigma}{1}$}  (s2)		
		
		(s1) edge	[loop above, out=120, in=70,looseness=5] node [left,xshift=-.1cm,yshift=-0.05cm]{$\LW{\Sigma}{2}$} (s1)
		(s2) edge  node[above] {$\LW{a}{3}$} (s3)
		(s2) edge  [out=-50, in=-130,looseness=.5] node[above,xshift=-0.6cm,yshift=-0.00cm] {$\LW{b}{1}$}(s4)
		
		(s3) edge	[loop above, out=120, in=70,looseness=5] node [right,xshift=.2cm,yshift=-0.05cm]{$\LW{\Sigma}{3}$} (s3)
		(s4) edge	[loop above, out=120, in=70,looseness=5] node [right,xshift=.2cm,yshift=-0.05cm]{$\LW{\Sigma}{1}$} (s4)
		;
		
		\node  (A2) [ below of=A, yshift=-0cm]{$\A_2$};
		\node[right of = A2,below of =A2, initial above, state,xshift=1.5cm,yshift=1cm] (s0) {$s_0$};
		\node[state] (s1) [ left of=s0] {$s_1$};
		\node[state] (s2) [ right of=s0] {$s_2$};
		\node[state] (s3) [right of=s2] {$s_3$};
		\node[state] (s4) [right of=s3] {$s_4$};
		
		\path 
		(s0) edge node [above]{$\Sigma$}   (s1)		
		(s0) edge node{$\Sigma$}  (s2)		
		
		(s1) edge[double]	[loop above, out=120, in=70,looseness=5] node [left,xshift=-.1cm,yshift=-0.05cm]{$\Sigma$} (s1)
		(s2) edge[double]  node[above] {$a$} (s3)
		(s2) edge  [out=-50, in=-130,looseness=.5] node[above,xshift=-0.6cm,yshift=-0.00cm] {$b$}(s4)
		
		(s3) edge[double]	[loop above, out=120, in=70,looseness=5] node [right,xshift=.1cm,yshift=0.1cm]{$\Sigma$} (s3)
		(s4) edge	[loop above, out=120, in=70,looseness=5] node [right,xshift=.1cm,yshift=0.1cm]{$\Sigma$} (s4)
		;
		
		\node  (A3) [ below of=A2, yshift=-0cm]{$\A_3$};
		\node[right of = A3,below of =A3, initial above, state,xshift=1.5cm,yshift=1cm] (s0) {$s_0$};
		\node[state] (s1) [ left of=s0] {$s_1$};
		\node[state] (s2) [ right of=s0] {$s_2$};
		\node[state] (s3) [right of=s2] {$s_3$};
		\node[state] (s4) [right of=s3] {$s_4$};
		
		\path 
		(s0) edge node [above]{$\Sigma$}   (s1)		
		(s0) edge node{$\Sigma$}  (s2)		
		
		(s1) edge	[loop above, out=120, in=70,looseness=5] node [left,xshift=-.1cm,yshift=-0.05cm]{$\Sigma$} (s1)
		(s2) edge[double]  node[above] {$a$} (s3)
		(s2) edge  [out=-50, in=-130,looseness=.5] node[above,xshift=-0.6cm,yshift=-0.00cm] {$b$}(s4)
		
		(s3) edge[double]	[loop above, out=120, in=70,looseness=5] node [right,xshift=.1cm,yshift=0.1cm]{$\Sigma$} (s3)
		(s4) edge	[loop above, out=120, in=70,looseness=5] node [right,xshift=.1cm,yshift=0.1cm]{$\Sigma$} (s4)
		;

	\end{tikzpicture}
	\caption{A $\LimSup$ automaton $\A$ and corresponding B\"uchi automata $\A_2$ and $\A_3$, as per \cref{cl:decomposeG2}. (Accepting transitions in $\A_2$ and $\A_3$ are marked with double lines.) Observe that $\A$ is not \HD and Eve loses the two-token game on $\A$, while both $\A_2$ and $\A_3$ are \HD. (In $\A$, if Eve goes from $s_0$ to $s_1$, Adam goes from $s_0$ to $s_2$ and continues with an $a$, and if she goes from $s_0$ to $s_2$, Adam goes from $s_0$ to $s_1$ and continues with a $b$. In $\A_2$ Eve goes from $s_0$ to $s_1$ and in $\A_3$ from $s_0$ to $s_2$.)}
	\label{fig:THDnotHD}
\end{figure}

\begin{thm}\label{cl:2TokenKToken}
	Given a quantitative automaton $\A$, if Eve wins $G_2(\A)$ then she also wins $G_k(\A)$ for any $k\in \mathbb{N}\setminus\{0\}$. Furthermore, if her winning strategy in $G_2(\A)$ has memory of size $m$ and $\A$ has $n$ states, then she has a winning strategy in $G_k(\A)$ with memory of size $n^{k-1} \cdot m^k$.
\end{thm}

\begin{proof}
	This is the generalisation of~\cite[Thm 14]{BK18}. The proof is similar to Bagnol and Kuperberg's original proof, but without assuming positional strategies for Eve in $G_k(\A)$. If Eve wins $G_2(\A)$ then she obviously wins $G_1(\A)$, using her $G_2$ strategy with respect to two copies of Adam's single token in $G_1$. We thus consider below $G_k(\A)$ for every $k\in\Nat\setminus\{0,1,2\}$.
	
	Let $s_2$ be a winning strategy for Eve in $G_2(\A)$. We inductively show that Eve has a winning strategy $s_i$ in $G_i(\A)$ for each finite $i$. To do so, we assume a winning strategy $s_{i-1}$ in $G_{i-1}(\A)$. The strategy $s_i$ maintains some additional (not necessarily finite) memory that maintains the position of one virtual token in $\A$, a position in the (not necessarily finite) memory structure of $s_{i-1}$, and a position in the (not necessarily finite) memory structure of $s_2$. The virtual token is initially at the initial state of $\A$. 
	The strategy $s_i$ then plays as follows: at each turn, after Adam has moved his $i$ tokens and played a letter (or, at the first turn, just played a letter), it first updates the $s_{i-1}$ memory structure, by ignoring the last of Adam's tokens, and, treating the position of the virtual token as Eve's token in  $G_{i-1}(\A)$, it updates the position of the virtual token according to the strategy $s_{i-1}$; it then updates the $s_2$ memory structure by treating Adam's last token and the virtual token as Adam's $2$ tokens in $G_2(\A)$, and finally outputs the transition to be played according to $s_2$.
	
	We now argue that this strategy is indeed winning in $G_{i}(\A)$.
	Since $s_{i-1}$ is a winning strategy in $G_{i-1}(\A)$, the virtual token traces a run of which the value is at least as large as the value of any of the runs built by the first $i-1$ tokens of Adam. Since $s_2$ is also winning, the value of the run built by Eve's token is at least as large as the values of the runs built by the virtual token and by Adam's last token. Hence, Eve is guaranteed to achieve at least the supremum value of Adam's $i$ runs, making this a winning strategy in $G_i(\A)$.
	
	As for the memory size of a winning strategy for Eve in $G_k(\A)$, let $m$ be the memory size of her winning strategy in $G_2(\A)$ and $n$ the number of states in $\A$. Then, by the above construction of her strategy in $G_k(\A)$, the memory of her strategy in $G_3(\A)$ is $n$ for the virtual token times $m$ for the copy of her  memory in $G_2(\A)$ times $m$ for the copy of her  memory in $G_{i-1}(\A)=G_2(\A)$, namely $n\cdot m \cdot m = n \cdot m^2$. Then for $G_4(\A)$ it is $n \cdot m \cdot (n \cdot m^2) = n^2 \cdot m^3$; for $G_5(\A)$ it is $n \cdot m \cdot (n^2 \cdot m^3) = n^3 \cdot m^4$, and for $G_k(\A)$ it is $n^{k-1} \cdot m^k$.
\end{proof}

We proceed with the definition of $k$-\HDness (a related, but different notion is the ``width'' of an automaton \cite{KM19}), based on the $k$-runs letter game (not to be confused with $G_k$, the $k$-token game), which generalises the letter game.

\begin{defi}[$k$-\HD and $k$-runs letter game]\label{def:kHD}
	A configuration of the game on a quantitative automaton $\A=(\Sigma,Q,\iota,\delta)$ is a tuple $q^k\in Q^k$ of states of $\A$, initialised to $\iota^k$.
	
	In a configuration $(q_{i,1},\dots,q_{i,k})$, the game proceeds as follows to the next configuration $(q_{i+1,1},\dots,q_{i+1,k})$.  
	\begin{itemize}
		\item Adam picks a letter $\sigma_{i}\in\Sigma$, then
		\item Eve chooses for each $q_{i,j}$, a  transition
		$\trans{q_{i,j}}{\sigma_{i}:\weight_{i,j}}{q_{i+1,j}}$
	\end{itemize}
	In the limit, a play consists of an infinite word $w$ that is derived from the concatenation of $\sigma_0,\sigma_1,\ldots$, as well as of $k$  infinite sequences $\rho_0,\rho_1,\ldots$ of transitions.  Eve wins the play if $\max_{j\in \{1\ldots k\}}\Val(\rho_j) = \A(w)$. 
	
	If Eve has a winning strategy, we say that $\A$ is $k$-\HD, or that $\HD_k(\A)$ holds.
\end{defi}
Notice that the standard letter game (\cref{def:HistoryDet}) is a $1$-run letter game and standard \HD (\cref{def:HistoryDet}) is $1$-\HD.

Next, we use $\HD_k(\A)$ to show that $G_2$ characterises \HDness.

\begin{propC}[\cite{BK18}]\label{cl:kHD-G2iffHD}
	Given a quantitative automaton $\A$, if $\HD_k(\A)$ for some $k\in \mathbb{N}$, and Eve wins $G_k(\A)$, then $\A$ is \HD.
\end{propC}

\begin{proof}
	The argument is identical to the one used in~\cite{BK18}, which we summarise here.
	The strategy $\tau$ for Eve in $\HD_k(\A)$ provides a way of playing $k$ tokens that guarantees that one of the $k$ runs formed achieves the automaton's value on the word $w$ played by Adam. If Eve moreover wins $G_k(\A)$ with some strategy $s_k$, she can, in order to win in the letter game, play $s_k$ against Adam's letters and $k$ virtual tokens that she moves according to $\tau$. The winning strategy $\tau$ guarantees that one of the $k$ runs built by the $k$ virtual tokens achieves $\Val(w)$; then her strategy $s_k$ guarantees that her run also achieves $\Val(w)$.
\end{proof}

It remains to prove that if Eve wins $G_2(\A)$, then $\HD_k(\A)$ for some $k$.

Given a $\LimSup$ automaton $\A$, with weights $\{1,\dots, k\}$, we define $k-1$ auxiliary B\"uchi automata $\A_2,\dots,\A_{k}$ with acceptance on transitions: each $\A_x$ is a copy of $\A$, where a transition is accepting if its weight $i$ in $\A$ is at least $x$. (See \cref{fig:THDnotHD}.) Dually,
given a $\LimInf$ automaton $\A$, each $A_x$ is a coB\"uchi automaton, and consists of a copy of $\A$ where transitions with weights smaller than $x$ are rejecting, while those with weights $x$ or larger are accepting. Again, in both cases $\A_x$ recognises the set of words $w$ such that $\A(w)\geq x$.

We now use these auxiliary automata to argue that if $G_2(\A)$ then $\HD_{k-1}(\A)$.

\begin{lem}\label{cl:decomposeG2}
	Given a $\LimSup$ or $\LimInf$ automaton $\A$ with weights $\{1,\dots, k\}$, if Eve wins $G_2(\A)$, then for all $x\in \{2,\dots, k\}$, Eve also wins $G_2(\A_x)$.
\end{lem}

\begin{proof}
	Since $\A_x$ is identical to $\A$ except for the acceptance condition or value function, Eve can use in $G_2(\A_x)$ her winning strategy in $G_2(\A)$.

	For the $\LimSup$ case, if one of Adam's runs sees an accepting transition infinitely often, the underlying transition of $\A$ visited infinitely often has weight at least $x$. Then, Eve's strategy guarantees that her run also sees infinitely often a value at least as large as $x$, corresponding to an accepting transition in $G_2(\A_x)$.
	
	Similarly, for the $\LimInf$ case, if one of Adam's runs avoids seeing a rejecting transition infinitely often in $\A_x$, then this run's value in $\A$ is at least $x$, and Eve's strategy guarantees that her run's value in $\A$ is at least $x$, meaning  that it avoids seeing a rejecting transition in $\A_x$ infinitely often, and accepts.
\end{proof}

\begin{lem}\label{cl:G2Gk}
	Given a $\LimSup$ or $\LimInf$ automaton $\A$ with weights $[1..k]$, if Eve wins $G_2(\A_x)$ for all $x\in [2..k]$ then $\HD_{k-1}(\A)$ holds.
\end{lem}

\begin{proof}
	From \cref{cl:decomposeG2}, if Eve wins $G_2(\A)$, then for all $x\in [2..k]$, Eve also wins $G_2(\A_x)$. 
	We first argue that each $\A_x$ is \HD.

	Since each $A_x$ is a B\"uchi or coB\"uchi automaton, this implies that for all $x\in [2..k]$, the automaton $A_x$ is \HD~\cite{BK18,BKLS20b}, witnessed by a winning strategy $s_x$ for Eve in the letter game on each $\A_x$.
	
	Now, in the $(k-1)$-run letter game on $\A$, Eve can use each $s_x$ to move one token. Then, if Adam plays a word $w$ with some value $\Val(w)=i$, this word is accepted by $\A_i$, and therefore the strategy $s_i$ guarantees that the run of the $i^{th}$ token achieves at least the value $i$, corresponding to seeing accepting transitions of $\A_i$ infinitely often for the $\LimSup$ case, or eventually avoiding rejecting transitions in the $\LimInf$ case.
\end{proof}

Finally, we combine the $G_2$ and $\HD_{k-1}$ strategies in $\A$ to show that $\A$ is \HD. 

\begin{thm}\label{cl:LimInfSupG2}
	A nondeterministic $\LimSup$ or $\LimInf$ automaton $\A$ is \HD if and only if Eve wins $G_2(\A)$.
\end{thm}
\begin{proof}
	If $\A$ is \HD then Eve can use the letter-game strategy to win in $G_2(\A)$, ignoring Adam's moves. 
	If Eve wins $G_2(\A)$ then by \cref{cl:decomposeG2} and \cref{cl:G2Gk} she wins $\HD_{k-1}(\A)$, where $k$ is the number of weights in $\A$. By \cref{cl:2TokenKToken} she also wins $G_{k-1}(\A)$ and, finally, by \cref{cl:kHD-G2iffHD} we get that $\A$ is \HD.
\end{proof}

We show next that $G_2$ characterises \HDness also for $\Sup$ automata, by reducing the problem to the case of $\LimSup$ automata.
\begin{cor}\label{cl:G2-oSup}
	A nondeterministic $\Sup$ automaton $\A$ on infinite words is \HD if and only if Eve wins $G_2(\A)$.
\end{cor}

\begin{proof}
	One direction is direct: if $\A$ is \HD, then Eve wins $G_2(\A)$ by using her \HD strategy and ignoring Adam's tokens.
	
	For the other direction,
	given a $\Sup$ automaton $\A$ with initial state $\iota$ and $k$ weights, we construct a $\LimSup$ automaton $\A'$, such that i) if $\A'$ is \HD then $\A$ is \HD, and ii) if Eve wins  $G_2(\A)$, then she also wins $G_2(\A')$. As by \cref{cl:LimInfSupG2} Eve wins $G_2(\A')$ if and only if $\A'$ is \HD, we conclude that if Eve wins $G_2(\A)$, then $\A$ is \HD.

	$\A'$ has the same alphabet as $\A$. The state-space of $\A'$ consists of $k$ copies of $\A$, one for each weight: for each state $q$ of $\A$, we have states $q_i$ of $\A'$, for $i\in [1..k]$. In the $i^{\mathit{th}}$ copy of $\A$, we only keep transitions of weight up to $i$ in $\A$, and all of them have weight $i$ in $\A'$. Transitions of higher weight $x$ move from the $i^{\mathit{th}}$ copy to the $x^{\mathit{th}}$ copy. More precisely, for each pair of weights $i$ and $x$:
	\begin{itemize}
		\item $(q_i, \sigma, i,q_i')$ is a transition in $\A'$ if $(q,\sigma,x,q')$ is a transition in $\A$ and $x\leq i$.
		\item $(q_i, \sigma, x, q_x')$ is a transition in $\A'$ if $(q,\sigma, x, q')$ is a transition in $\A$ and $x>i$.
	\end{itemize}
	The initial state of $\A'$ is $\iota_1$.
	
	First note that this $\LimSup$-automaton computes the same function as $\A$: runs of $\A$ and $\A'$ are in value-preserving bijection since in $\A'$ the weight that occurs on each transition is the largest weight seen so far in the corresponding run of $\A$.
	
	Then, observe that if $\A'$ is \HD then $\A$ is \HD. 
	Indeed, Eve's winning strategy $s'$ in the letter game on $\A'$ can be used in the letter game on $\A$ by choosing the transition $(q,\sigma, x, q')$ instead of $(q_i,\sigma, x, q_x')$ and the available transition $(q, \sigma, x, q')$ maximising   $x$ instead of $(q_i,\sigma, i, q_i')$. Again, this strategy achieves a run of the same value as $s'$ on every word since whenever $s'$ moves to the next copy of $\A$, it sees a transition of the corresponding weight, guaranteeing at least this value. Eve therefore also wins in the letter game on $\A$.

	Similarly, we claim that if Eve wins $G_2(\A)$, she also wins $G_2(\A')$. Indeed, given a winning strategy $s$ in $G_2(\A)$, she can play according to the strategy $s'$ that chooses the available $(q_i,\sigma,i,q_x')$ or $(q_i, \sigma, x, q_x')$ transition whenever $s$ chooses $(q,\sigma,x,q')$, and interprets Adam's moves $(q_i,\sigma, x, q_x')$ in $G_2(\A')$ as  $(q,\sigma, x, q')$ in $G_2(\A)$ and Adam's moves $(q_i,\sigma, i, q_i')$ in $G_2(\A')$ as the available transition $(q, \sigma, x, q')$ maximising  $x$ in $G_2(\A)$. Then, whatever values Adam's tokens achieve in $G_2(\A')$, $s$ guarantees that Eve's token in $G_2(\A)$ does better, and $s'$ guarantees this same value for Eve's token in $G_2(\A')$. 
\end{proof}

\begin{rem}
	An alternative proof for the above corollary can be obtained by following the same proof schema as for \cref{cl:LimInfSupG2}, but decomposing the $\Sup$ automaton into $\Reach$ automata instead of B\"uchi or coB\"uchi automata.
\end{rem}

\subsection{Solving $G_2$ and Deciding \HDness}\label{sec:solveG2}
\hfill\\[-6pt]

\noindent
Now that we have established that $G_2$ characterises history-determinism for $\Sup$, $\LimSup$, and $\LimInf$ automata, we study the complexity of solving $G_2$ in each case. We start with the case of $\Sup$ automata on infinite words, which is the simplest, via reduction to polynomially sized coB\"uchi game.

\begin{thm}\label{cl:solveG2-Sup}
	Deciding whether a $\Sup$ automaton $\A$ on infinite words is \HD can be done in time polynomial in the size of $\A$.
\end{thm}

\begin{proof}
	By \cref{cl:G2-oSup}, it is enough to solve $G_2(\A)$.
	
	We encode $G_2(\A)$ for a $\Sup$ automaton $\A=(\Sigma,Q,\iota,\delta)$ as a coB\"uchi game as follows. The arena is that of $G_2(\A)$, represented as the product of the alphabet and three copies of $\A$, to reflect the current letter and the current position of each of the three runs, and a variable indicating whose turn it is to move. In addition, a variable $x_E$ keeps track of the greatest value seen by Eve's run so far,  and $x_A$ keeps track of the greatest value seen by either of Adam's runs so far.  The winning condition for Eve  is then that eventually $x_E\geq x_A$ remains true, which is a coB\"uchi condition.
	The size of the arena is thus polynomial, and the complexity of solving coB\"uchi games is quadratic~\cite{CH12}.
\end{proof}

We proceed with the $\LimSup$ and $\LimInf$ cases, for which $G_2$ can be solved via a reduction to a parity game. The $G_2$ winning condition for $\LimSup$ automata can be encoded by adding carefully chosen priorities to the arena of $G_2(\A)$, while for $\LimInf$ the encoding requires additional positions.

\begin{thm}\label{cl:solveG2LimSup}
	Deciding whether a $\LimSup$ automaton $\A$ of size $n$ with $k$ weights is \HD is quasipolynomial in $n$, and if $k$ is in $O(\log n)$, in time polynomial in $n$.
\end{thm}

\begin{proof}
	By \cref{cl:LimInfSupG2}, it is enough to solve $G_2(\A)$.
	
	We encode the game $G_2(\A)$, for a $\LimSup$ automaton $\A=(\Sigma,Q,\iota,\delta)$, into a parity game as follows. The arena is simply the arena of $G_2(\A)$, seen as a product of the alphabet and three copies of $\A$, to reflect the current letter and the current position of each of the three runs (one for Eve, two for Adam), and a variable to keep track of which turn it is to play. (Adam's turn to give a letter, Eve's turn to choose a transition, Adam's turn to choose his first transition, or Adam's turn to choose his second transition.)
	
	Adam's letter-picking moves are labelled with priority $0$, Eve's choices of transition $\trans{q}{\sigma:\weight}{q'}$ are labelled with priority $2\weight$ and Adam's choices of transition $\trans{q}{\sigma:\weight}{q'}$  are labelled with priority $2\weight -1$. 
	
	We claim that Eve wins this parity game if and only if she wins $G_2(\A)$, that is, the priorities correctly encode the winner of $G_2(\A)$. Observe that the even priorities seen infinitely often in a play of the parity game are exactly priorities $2x$, where $x$ is a weight seen infinitely often in Eve's run in the corresponding play in $G_2(\A)$. The odd priorities seen infinitely often on the other hand are $2x-1$, where $x>0$ occurs infinitely often on one of Adam's runs in the corresponding play of $G_2(\A)$. Hence, Eve can match the maximal value of Adam's runs in $G_2(\A)$ if and only if she can win the parity game that encodes $G_2(\A)$.
	
	The number of positions in this game is polynomial in the size $n$ of $\A$; the maximal priority is linear in the number $k$ of its weights. It can be solved in quasipolynomial time, or in polynomial time if $k$ is in $O(\log n)$, using the reader's favourite state-of-the-art parity game algorithm, for instance~\cite{CJKLS17}.
\end{proof}

\begin{thm}\label{cl:solveG2liminf}
	Deciding whether a $\LimInf$ automaton $\A$ of size $n$ with $k$ weights is \HD can be done in time exponential in $k$, and if $k$ is in $O(\log n)$, in time polynomial in $n$.
\end{thm}

\begin{proof}
	By \cref{cl:LimInfSupG2}, it is enough to solve $G_2(\A)$.
	
	As in the proof of \cref{cl:solveG2LimSup}, we can represent $G_2(\A)$ as a game on an arena that is the product of three copies of $\A$, one for Eve and two for Adam. The winning condition for Eve is that the smallest weight seen infinitely often on the run built on her copy of $\A$ should be at least as large as both of the minimal weights seen infinitely often on the runs built on Adam's copies.
	We will encode this winning condition as a parity condition, but, unlike in the $\LimSup$ case, we will need to use an additional memory structure, which we describe now.
	
	The weights on Eve's run will be encoded by \textit{odd} priorities, with smaller weights corresponding to higher priorities, as for $\LimInf$ the lowest weight seen infinitely often is the one that matters, while weights on Adam's runs will be encoded by \textit{even} priorities, but only once both of Adam's runs have seen the corresponding weight or a lower one.
	Keeping track of this last condition requires the following additional memory structure, which encodes which of Adam's runs has seen which weight recently.

\begin{figure}[h]
    \centering
    \begin{tikzpicture}[thick,initial text=, every initial by arrow/.style={|->}]
    \tikzset{every state/.style = {minimum size =10}}

    \node[state, initial below] (0) at (0,0) {$x_5=0$};
    \node[state] (1) at (4,0) {$x_5=1$};
    \node[state] (2) at (-4,0) {$x_5=2$};

    \path[-stealth]
     (0) edge node[above] {$\tau_1$} (1)
     (0) edge[loop above] node[right,xshift=0.1cm] { $\setminus \{\tau_1, \tau_2 \}$} ()
     (1) edge[loop above] node[right] {$\setminus \{\tau_2\}$} ()
     (0) edge node[below] {$\tau_2$} (2)
     (2) edge[loop above] node[right, xshift=0.1cm] {$\setminus\{\tau_1\}$} ()
     (1) edge[bend left] node[below] {$\tau_2$} (0)
     (2) edge[bend left] node[above] {$\tau_1$} (0)
     ;
    \end{tikzpicture}
 \caption{Memory structure as described in the proof of \cref{cl:solveG2liminf}
 	for keeping track of when both of Adam's runs in $G_2(\A)$ have seen a priority smaller than $i=5$. The labels $\tau_1$ and $\tau_2$ stand for moves in $G_2(\A)$ that correspond to Adam choosing a transition of $\A$ with weight $w\leq 5$ in his first and second runs, respectively. A label $\setminus \{\ldots\}$ stands for all moves in $G_2(\A)$ except for those in $\{\ldots\}$.}
    \end{figure}
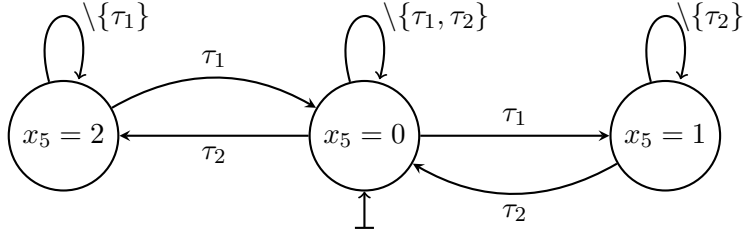

	More precisely, let $k$ be the number of weights in $\A$. Moves corresponding to Eve choosing a transition of $\A$ with weight $i$ have priority $2(k-i+1)-1$, that is, an odd priority that is larger the smaller $i$ is. Further, for each weight $i\in[1..k]$, we use a three-valued variable $x_i\in\{0,1,2\}$, initiated to $0$, which gets updated as follows: if $x_i=0$ and the game move corresponds to Adam choosing a transition of $\A$ with weight $w\leq i$ in one of his runs, $x_i$ is updated to $1$ or $2$ according to which of Adam's run saw this weight; if $x_i=1$ (resp.\ $2$) and the game move corresponds to Adam choosing a transition of $\A$ with weight $w\leq i$ in his second (resp.\ first) run, then $x_i$ is reset to $0$. Moves that reset variables to $0$ have priority $2(k-i+1)$ for the minimal $i$ such that the transition resets $x_i$ to $0$. 
	Other moves have priority~$1$.
	
	We now argue that the highest priority seen infinitely often along a play is even if and only if the $\LimInf$ value of Eve's run is at least as high as that of both of Adam's runs.
	Indeed, the maximal odd priority seen infinitely often on a play is $2(k-i+1)-1$ such that $i$ is the minimal priority seen on Eve's run infinitely often, and the maximal even priority seen infinitely often is $2(k-j+1)$ where $j$ is the minimal weight such that both of Adam's runs see $j$ or a smaller priority infinitely often. In particular, $2(k-i+1)-1<2(k-j+1)$ if and only if $i\geq j$, that is, if Eve wins $G_2(\A)$.
	
	This parity game is of size polynomial in $n$  and exponential in $k$, where the latter is due to the memory structure ($\{0,1,2\}^k$) and has $2k$ priorities. As the number of priorities is logarithmic in the size of the game, it can be solved in polynomial time~\cite{CJKLS17}. If $k$ is in $O(\log n)$, then the algorithm is polynomial in the size $n$ of $\A$.
\end{proof}

These results can be compared to the problem referred to as zero-regret synthesis against word-strategies in~\cite{HPR17}. There, the $\LimInf$ and $\Sup$ cases are solved in exponential time in the size of the automaton via determinisation. In contrast, our procedure is polynomial for $\Sup$, quasipolynomial for $\LimSup$  and exponential in the number of weights for $\LimInf$.\\

In contrast to the cases considered in \cref{sec:G1}, where strategies in $G_1$ induce \HD strategies of the same memory size, a winning $G_2$ strategy does not necessarily induce an \HD strategy of the same memory size (even when it implies the existence of such a strategy). We now analyse the size of the $\HD$ strategies which our proofs show exist whenever Eve wins $G_2$, and discuss the implications for the determinisability of $\HD$ $\LimSup$ automata. 

\begin{cor}\label{cl:MemoryBounds}
	Consider an $\HD$ $\Sup$, $\LimSup$, or  $\LimInf$  automaton $\A$ of size $n$ with $k+1$ weights on infinite words. If $\A$ is a $\Sup$ or $\LimSup$ automaton, there is an $\HD$ strategy for $\A$ with memory polynomial in $n$ if $k$ is in $O(\log n)$, and of memory exponential in $k$ otherwise. If $\A$ is a $\LimInf$ automaton, there is an $\HD$ strategy for $\A$ with memory exponential in $n$.
\end{cor}

\begin{proof}
	We construct an $\HD$ strategy for $\LimSup$ or $\LimInf$ $\A$, by combining an $\HD_k$ strategy and a $G_k$ strategy for it. 
	
	The $\HD_k$ strategy---which, like the $\HD$ strategy, is hard to compute directly---combines the $\HD$ strategies of the $k$ auxiliary B\"uchi or coB\"uchi automata for $\A$, as constructed in \cref{cl:decomposeG2}. For $\HD$ B\"uchi automata, which are equivalent to deterministic automata of quadratic size~\cite{KS15}, there always exist polynomial resolvers: indeed, the letter game can be represented as a polynomial parity game, in which a positional strategy for Eve corresponds to a resolver.
	For $\HD$ coB\"uchi automata on the other hand, these auxiliary strategies might have exponential memory in the number of states of $\A$~\cite{KS15}. 
	
	The $G_k$ strategy on the other hand is positional for $\LimSup$, since it can be encoded as a parity game directly on the $G_k(\A)$ arena, similarly to the reduction in~\cref{cl:solveG2LimSup}; the size of the $G_k(\A)$ arena is $O(n^{k+1})$. The overall $\HD$ strategy for $\LimSup$ therefore needs memory exponential in the number of weights.
	
	For $\LimInf$ on the other hand, by \cref{cl:solveG2liminf,cl:2TokenKToken}, the $G_k$ strategy can do with memory of size $n^{k-1} \cdot 3^{k^2}$. The overall $\HD$ strategy however needs memory exponential in both $n$ and $k$, due to strategies in the component coB\"uchi automata potentially requiring memory exponential in $n$.

	For the $\Sup$ case, recall the $\LimSup$ automaton $\A'$, which has $k*n$ states and $k$ weights, from the proof of~\cref{cl:G2-oSup}, which is \HD if and only if $\A$ is \HD. From the same proof, an \HD strategy in $\A'$ can be interpreted as an \HD strategy in $\A$. Since $\A'$ is of polynomial size in $n$ and has as many weights as $\A$, $\A'$ has an \HD strategy that has memory exponential in $k$, and therefore so does $\A$. If $k$ is in $O(\log n)$, then the memory is polynomial in $n$.
\end{proof}

We leave open whether the memory bounds of \cref{cl:MemoryBounds} can be improved upon. Already for coB\"uchi automata, it is known that deciding whether an automaton is $\HD$ is polynomial despite there being automata for which the optimal $\HD$ strategy needs exponential memory. Hence, at least for the $\LimInf$ case, we cannot expect to do  much better. However, for the $\Sup$ and $\LimSup$ cases, it might be that strategies with polynomial memory suffice.

Our proof does however imply that if a $\LimSup$ automaton $\A$ is $\HD$, then there is a \textit{finite memory} $\HD$ strategy, which implies that $\A$ is determinisable, without increasing the number of weights, by taking a product of $\A$ with the finite $\HD$ strategy. (Recall that every $\LimInf$ automaton can be determinised,  while not every $\LimSup$ automaton can.)

\begin{cor}
	Every \HD $\LimSup$ automaton is equivalent to a deterministic one with at most an exponential number of states and the same set of weights.
\end{cor}

\section{Conclusions}
We have extended the token-game approach to characterising history-determinism from the Boolean ($\omega$-regular) to the quantitative setting. 
Already $1$-token games turn out to be useful for characterising history-determinism for some quantitative automata. For $\LimSup$, $\LimInf$ and $\Sup$ automata on infinite words, one token is not enough, but the $2$-token game does the trick. 
Given the correspondence between deciding history-determinism and the best-value synthesis problem, our results also directly provide algorithms both to decide whether the synthesis problem is realisable and to compute a solution strategy.

This application further motivates understanding the limits of these techniques.
Whether the $2$-token game $G_2$ characterises more general Boolean classes of automata beyond B\"uchi and coB\"uchi automata is already an open question. 
Similarly, we leave open whether the $G_2$ game also characterises history-determinism for limit-average automata and other quantitative automata. At the moment we are not aware of examples of automata of any kind (quantitative, pushdown, register, timed, \ldots) for which Eve could win $G_2$ despite the automaton not being history-deterministic, yet even for parity automata, a proof of characterisation remains elusive.

The case of limit average automata is particularly interesting as it is a well-studied value-function. However, deciding the \HDness of these automata presents some additional difficulties: first, the tokens games are less straight-forward to solve, as their winning conditions can no longer be encoded into parity conditions, as is the case for $\LimSup$ and $\LimInf$; second, the characterisation of \HDness is also likely to be challenging, as the techniques we have used here (present-focused value functions, decomposition into auxiliary automata) don't seem to apply.

\section*{Acknowledgment}
  \noindent We thank Guillermo A. P{\'{e}}rez for discussing history-determinism of discounted-sum 
  and limit-average automata.

\bibliographystyle{alphaurl}
\bibliography{gfg}

\newcommand{\etalchar}[1]{$^{#1}$}
\begin{thebibliography}{CJK{\etalchar{+}}17}

\bibitem[AKL10]{AKL10}
Benjamin Aminof, Orna Kupferman, and Robby Lampert.
\newblock Reasoning about online algorithms with weighted automata.
\newblock {\em {ACM} Trans. Algorithms}, 6(2):28:1--28:36, 2010.
\newblock \href {https://doi.org/10.1145/1721837.1721844}
  {\path{doi:10.1145/1721837.1721844}}.

\bibitem[And06]{And06}
Daniel Andersson.
\newblock An improved algorithm for discounted payoff games.
\newblock In {\em Proc.\ of ESSLLI Student Session}, pages 91--98, 2006.

\bibitem[BH21]{BH21}
Udi Boker and Guy Hefetz.
\newblock Discounted-sum automata with multiple discount factors.
\newblock In {\em Proceedings of {CSL}}, pages 12:1--12:23, 2021.
\newblock \href {https://doi.org/10.4230/LIPIcs.CSL.2021.12}
  {\path{doi:10.4230/LIPIcs.CSL.2021.12}}.

\bibitem[BK18]{BK18}
Marc Bagnol and Denis Kuperberg.
\newblock B{\"u}chi good-for-games automata are efficiently recognizable.
\newblock In {\em Proceedings of {FSTTCS}}, 2018.
\newblock \href {https://doi.org/10.4230/LIPIcs.FSTTCS.2018.16}
  {\path{doi:10.4230/LIPIcs.FSTTCS.2018.16}}.

\bibitem[BKLS20]{BKLS20b}
Udi Boker, Denis Kuperberg, Karoliina Lehtinen, and Micha{\l} Skrzypczak.
\newblock On succinctness and recognisability of alternating good-for-games
  automata.
\newblock {\em arXiv preprint arXiv:2002.07278}, 2020.

\bibitem[BKS17]{BKS17}
Udi Boker, Orna Kupferman, and Micha{\l} Skrzypczak.
\newblock How deterministic are good-for-games automata?
\newblock In {\em Proceedings of {FSTTCS}}, pages 18:1--18:14, 2017.
\newblock \href {https://doi.org/10.4230/LIPIcs.FSTTCS.2017.18}
  {\path{doi:10.4230/LIPIcs.FSTTCS.2017.18}}.

\bibitem[BL19]{BL19}
Udi Boker and Karoliina Lehtinen.
\newblock Good for games automata: From nondeterminism to alternation.
\newblock volume 140 of {\em LIPIcs}, pages 19:1--19:16, 2019.
\newblock \href {https://doi.org/10.4230/LIPIcs.CONCUR.2019.19}
  {\path{doi:10.4230/LIPIcs.CONCUR.2019.19}}.

\bibitem[BL21]{BL21}
Udi Boker and Karoliina Lehtinen.
\newblock History determinism vs. good for gameness in quantitative automata.
\newblock In {\em Proc.\ of {FSTTCS}}, pages 35:1--35:20, 2021.
\newblock \href {https://doi.org/10.4230/LIPIcs.FSTTCS.2021.38}
  {\path{doi:10.4230/LIPIcs.FSTTCS.2021.38}}.

\bibitem[BL22]{BL22}
Udi Boker and Karoliina Lehtinen.
\newblock Token games and history-deterministic quantitative automata.
\newblock In {\em {FOSSACS}}, pages 120--139, 2022.
\newblock \href {https://doi.org/10.1007/978-3-030-99253-8\_7}
  {\path{doi:10.1007/978-3-030-99253-8\_7}}.

\bibitem[Bok18]{Bok18}
Udi Boker.
\newblock Why these automata types?
\newblock In {\em Proceedings of {LPAR}}, pages 143--163, 2018.
\newblock \href {https://doi.org/10.29007/c3bj} {\path{doi:10.29007/c3bj}}.

\bibitem[Bok21]{Bok21}
Udi Boker.
\newblock Quantitative vs. weighted automata.
\newblock In {\em Proc.\ of {Reachbility Problems}}, pages 1--16, 2021.
\newblock \href {https://doi.org/10.1007/978-3-030-89716-1\_1}
  {\path{doi:10.1007/978-3-030-89716-1\_1}}.

\bibitem[CDH10]{CDH10}
Krishnendu Chatterjee, Laurent Doyen, and Thomas~A. Henzinger.
\newblock Quantitative languages.
\newblock {\em {ACM} Trans. Comput. Log.}, 11(4):23:1--23:38, 2010.
\newblock \href {https://doi.org/10.1145/1805950.1805953}
  {\path{doi:10.1145/1805950.1805953}}.

\bibitem[CF19]{CF19}
Thomas Colcombet and Nathana{\"{e}}l Fijalkow.
\newblock Universal graphs and good for games automata: New tools for infinite
  duration games.
\newblock In {\em Proc.\ of {FOSSACS}}, 2019.
\newblock \href {https://doi.org/10.1007/978-3-030-17127-8\_1}
  {\path{doi:10.1007/978-3-030-17127-8\_1}}.

\bibitem[CH12]{CH12}
Krishnendu Chatterjee and Monika Henzinger.
\newblock An {$O(n^2)$} time algorithm for alternating {B}{\"u}chi games.
\newblock In {\em Proceedings of ACM-SIAM symposium on Discrete Algorithms},
  pages 1386--1399. SIAM, 2012.
\newblock \href {https://doi.org/10.1137/1.9781611973099.109}
  {\path{doi:10.1137/1.9781611973099.109}}.

\bibitem[CJK{\etalchar{+}}17]{CJKLS17}
Cristian~S Calude, Sanjay Jain, Bakhadyr Khoussainov, Wei Li, and Frank
  Stephan.
\newblock Deciding parity games in quasipolynomial time.
\newblock In {\em Proceedings of {STOC}}, pages 252--263, 2017.
\newblock \href {https://doi.org/10.1145/3055399.3055409}
  {\path{doi:10.1145/3055399.3055409}}.

\bibitem[Col09]{Col09}
Thomas Colcombet.
\newblock The theory of stabilisation monoids and regular cost functions.
\newblock In {\em Proceedings of {ICALP}}, pages 139--150, 2009.
\newblock \href {https://doi.org/10.1007/978-3-642-02930-1\_12}
  {\path{doi:10.1007/978-3-642-02930-1\_12}}.

\bibitem[FJL{\etalchar{+}}17]{FJLPR17}
Emmanuel Filiot, Isma{\"{e}}l Jecker, Nathan Lhote, Guillermo~A. P{\'{e}}rez,
  and Jean{-}Fran{\c{c}}ois Raskin.
\newblock On delay and regret determinization of max-plus automata.
\newblock In {\em {LICS}}, pages 1--12, 2017.
\newblock \href {https://doi.org/10.1109/LICS.2017.8005096}
  {\path{doi:10.1109/LICS.2017.8005096}}.

\bibitem[FLW20]{FLW20}
Emmanuel Filiot, Christof L{\"{o}}ding, and Sarah Winter.
\newblock Synthesis from weighted specifications with partial domains over
  finite words.
\newblock In {\em Proceedings of {FSTTCS}}, 2020.
\newblock \href {https://doi.org/10.4230/LIPIcs.FSTTCS.2020.46}
  {\path{doi:10.4230/LIPIcs.FSTTCS.2020.46}}.

\bibitem[GJLZ21]{GJLZ21}
Shibashis Guha, Isma{\"e}l Jecker, Karoliina Lehtinen, and Martin Zimmermann.
\newblock A bit of nondeterminism makes pushdown automata expressive and
  succinct.
\newblock pages 53:1--53:20, 2021.
\newblock \href {https://doi.org/10.4230/LIPIcs.MFCS.2021.53}
  {\path{doi:10.4230/LIPIcs.MFCS.2021.53}}.

\bibitem[HMS]{HMS06}
Malte Helmert, Robert Mattm{\"{u}}ller, and Sven Schewe.
\newblock Selective approaches for solving weak games.
\newblock In Susanne Graf and Wenhui Zhang, editors, {\em Proceedings of
  {ATVA}}.
\newblock \href {https://doi.org/10.1007/11901914\_17}
  {\path{doi:10.1007/11901914\_17}}.

\bibitem[HP06]{HP06}
Thomas Henzinger and Nir Piterman.
\newblock Solving games without determinization.
\newblock In {\em Proceedings of {CSL}}, pages 395--410, 2006.
\newblock \href {https://doi.org/10.1007/11874683\_26}
  {\path{doi:10.1007/11874683\_26}}.

\bibitem[HPR16]{HPR16}
Paul Hunter, Guillermo~A. P{\'{e}}rez, and Jean{-}Fran{\c{c}}ois Raskin.
\newblock Minimizing regret in discounted-sum games.
\newblock In Jean{-}Marc Talbot and Laurent Regnier, editors, {\em {CSL}},
  volume~62 of {\em LIPIcs}, pages 30:1--30:17, 2016.
\newblock \href {https://doi.org/10.4230/LIPIcs.CSL.2016.30}
  {\path{doi:10.4230/LIPIcs.CSL.2016.30}}.

\bibitem[HPR17]{HPR17}
Paul Hunter, Guillermo~A. P{\'{e}}rez, and Jean{-}Fran{\c{c}}ois Raskin.
\newblock Reactive synthesis without regret.
\newblock {\em Acta Informatica}, 54(1):3--39, 2017.
\newblock \href {https://doi.org/10.1007/s00236-021-00410-0}
  {\path{doi:10.1007/s00236-021-00410-0}}.

\bibitem[KS15]{KS15}
Denis Kuperberg and Micha{\l} Skrzypczak.
\newblock On determinisation of good-for-games automata.
\newblock In {\em Proceedings of {ICALP}}, pages 299--310, 2015.

\bibitem[LR13]{LR13}
Christof L{\"{o}}ding and Stefan Repke.
\newblock Decidability results on the existence of lookahead delegators for
  {NFA}.
\newblock In {\em Proc.\ of {FSTTCS} 2013}, pages 327--338, 2013.
\newblock \href {https://doi.org/10.4230/LIPIcs.FSTTCS.2013.327}
  {\path{doi:10.4230/LIPIcs.FSTTCS.2013.327}}.

\bibitem[LZ22]{LZ20}
Karoliina Lehtinen and Martin Zimmermann.
\newblock Good-for-games {\(\omega\)}-pushdown automata.
\newblock volume~18, 2022.
\newblock \href {https://doi.org/10.46298/lmcs-18(1:3)2022}
  {\path{doi:10.46298/lmcs-18(1:3)2022}}.

\bibitem[MK19]{KM19}
Anirban Majumdar and Denis Kuperberg.
\newblock Computing the width of non-deterministic automata.
\newblock {\em Logical Methods in Computer Science}, 15, 2019.
\newblock \href {https://doi.org/10.23638/LMCS-15(4:10)2019}
  {\path{doi:10.23638/LMCS-15(4:10)2019}}.

\bibitem[Sha53]{Sha53}
L.~S. Shapley.
\newblock Stochastic games.
\newblock In {\em Proc.\ of Nat. Acad. Sci.}, volume~39, pages 1095--1100,
  1953.

\bibitem[ZP95]{ZP95}
Uri Zwick and Mike Paterson.
\newblock The complexity of mean payoff games on graphs.
\newblock {\em Electron. Colloquium Comput. Complex.}, 2(40), 1995.
\newblock \href {https://doi.org/10.1016/0304-3975(95)00188-3}
  {\path{doi:10.1016/0304-3975(95)00188-3}}.

\end{thebibliography}

\end{document}